\documentclass{article}
\pdfoutput=1
\usepackage[utf8]{inputenc}
\usepackage{jheppub}
\usepackage[export]{adjustbox}
\usepackage{amsmath,amsthm,amsfonts,amssymb,amscdx,mathrsfs}
\usepackage{scalerel}
\usepackage{multirow}
\usepackage{longtable}
\usepackage{mathtools}
\usepackage{upgreek}
\usepackage{subcaption}
\usepackage{tikz}
\usepackage{physics}
\usetikzlibrary{calc}   
\usetikzlibrary{math}   

\preprint{BRX-TH-6721}

\usetikzlibrary{intersections} 

\DeclareMathAlphabet{\mathpzc}{OT1}{pzc}{m}{it}
\DeclareMathAlphabet\mathsf{OT1}{lcmss}{m}{n}
\usepackage[mathscr]{eucal}

\newcommand{\be}{\begin{equation}}
\newcommand{\ee}{\end{equation}}

\newtheorem{theorem}{Theorem}[section]
\newtheorem{lemma}[theorem]{Lemma}
\newtheorem{corollary}[theorem]{Corollary}

\title{\boldmath Testing holographic entropy inequalities in $2+1$ dimensions}

\author[a]{Brianna Grado-White,}
\author[a]{Guglielmo Grimaldi,}
\author[a]{Matthew Headrick,} 
\author[b]{and Veronika E. Hubeny}

\affiliation[a]{Martin Fisher School of Physics, Brandeis University, Waltham MA 02453, USA}
\affiliation[b]{Center for Quantum Mathematics and Physics (QMAP)\\ Department of Physics \& Astronomy, University of California, Davis CA, USA}
\emailAdd{bgradowhite@brandeis.edu}
\emailAdd{ggrimaldi@brandeis.edu}
\emailAdd{headrick@brandeis.edu}
\emailAdd{veronika@physics.ucdavis.edu}

\abstract{We address the question of whether holographic entropy inequalities obeyed in static states (by the RT formula) are always obeyed in time-dependent states (by the HRT formula), focusing on the case where the bulk spacetime is $2+1$ dimensional. An affirmative answer to this question was previously claimed by Czech-Dong. We point out an error in their proof when the bulk is multiply connected. We nonetheless find strong support, of two kinds, for an affirmative answer in that case.  We extend the Czech-Dong proof for simply-connected spacetimes to spacetimes with $\pi_1=\mathbb{Z}$ (i.e.\ 2-boundary, genus-0 wormholes). Specializing to vacuum solutions,  we also numerically test thousands of distinct inequalities (including all known RT inequalities for up to 6 regions) on millions of randomly chosen configurations of regions and bulk spacetimes, including three different multiply-connected topologies; we find no counterexamples. In an appendix, we prove some (dimension-independent) facts about degenerate HRT surfaces and symmetry breaking. 
\newline
\newline
A video abstract is available at \url{https://www.youtube.com/watch?v=ols92YU8rus}.
}

\begin{document}
\maketitle
\flushbottom

\section{Introduction}

One of the major themes of holography has been the deep and pervasive connection between entanglement and spacetime geometry. This idea has its most concrete manifestation in the Ryu-Takayanagi (RT) formula, which states that in a holographic duality the entanglement entropy $S(A)$ of a boundary subregion $A$ of a static (or time-reflection invariant) time slice is given by the area of the minimal surface homologous to $A$ \cite{Ryu:2006ef}. The study of  entropy inequalities, begun as a check of validity for the RT proposal, has become an important arena for exploring the entanglement structure of holographic states.  In particular, entropy inequalities help to classify which states admit classical bulk duals: it has been shown that the RT formula (and therefore static holographic states) obey an infinite set of additional inequalities not obeyed by general quantum states \cite{Hayden:2011ag,Bao:2015bfa}. 

The first and simplest example of such an inequality is called the monogamy of mutual information (MMI)  \cite{Hayden:2011ag}, which reads
\be\label{eq:mmi}
S(AB) + S(BC) + S(AC) \geq S(A) + S(B) + S(C) + S(ABC)\,.
\ee
MMI is in fact a part of an infinite family of inequalities invariant under a dihedral action on the set of regions, the next member of which is the following 5-party inequality:
\begin{align}\label{eq:5-party}
\begin{split}
S(ABC) + &S(BCD) + S(CDE) + S(DEA) + S(EAB)\\
& \geq  S(AB) + S(BC) + S(CD) + S(DE) + S(EA) + S(ABCDE)\,.
\end{split}
\end{align}
Such holographic entropy inequalities together define the RT entropy cone in entropy space.\footnote{All the known inequalities obeyed by RT have been proven using so-called contraction maps \cite{Bao:2015bfa}. We will assume that this is a general fact (as argued in \cite{Bao:2024obe}), and somewhat sloppily use ``RT inequality'' to mean ``inequality provable by a contraction map''.} Despite much progress in recent years in enumerating and characterizing these inequalities, the full set remains unknown.\footnote{The program of elucidating the RT entropy cone for $\mathsf{N}$ parties \cite{Bao:2015bfa,Marolf:2017shp,Hubeny:2018trv,Hubeny:2018ijt,HernandezCuenca:2019wgh,He:2019ttu,He:2020xuo,Avis:2021xnz,Czech:2022fzb,Hernandez-Cuenca:2023iqh,Czech:2023xed} has proceeded along several distinct but interweaving routes, some focusing on characterizing its extreme rays and  others on uncovering its facets (defined by inequalities).  Since both the number of inequalities and the number of extreme rays grows very rapidly with $\mathsf{N}$, to date, we have complete knowledge only up to $\mathsf{N}=5$.   Nevertheless, thousands of extreme ray and inequality orbits are known for $\mathsf{N}=6$ \cite{Hernandez-Cuenca:2023iqh, hecdata}, as well as 2 infinite families for arbitrarily high $\mathsf{N}$ \cite{Czech:2023xed}. In building up towards a better eventual understanding of the RT cone, some works considered utilizing a more primal construct, the subaddivity cone, to reconstruct the RT cone extreme rays \cite{Hernandez-Cuenca:2019jpv,Hernandez-Cuenca:2022pst,He:2022bmi,He:2023cco,He:2023aif}, some examined the structural aspects or proof methods \cite{Cui:2018dyq,Li:2022jji,Czech:2024rco,Bao:2024obe}, while others explored various generalizations of the RT cone \cite{Bao:2020zgx,Walter:2020zvt,Bao:2020mqq,Akers:2021lms,Bao:2021gzu,He:2023rox} or focused on simpler structure afforded by averaging entropies \cite{Czech:2021rxe,Fadel:2021urx}.} Furthermore, little is known about their interpretation from a quantum information perspective, in other words, what they tell us about the qualitative nature of holographic states, as opposed to general quantum states. (However, for a conjectured  
interpretation of MMI, see \cite{Cui:2018dyq}, and for a generalization to higher parties, see \cite{Czech:2021rxe}.)

On the other hand, much less is known about properties of entropies in time-dependent states, for which the covariant prescription by Hubeny-Rangamani-Takayanagi (HRT) is required \cite{Hubeny:2007xt}. In particular, whether the HRT formula is constrained by the same set of inequalities as the RT formula remains an open question. This question goes to the heart of the physical meaning and implications of the inequalities. In the RT setting, perhaps surprisingly, the proofs of the inequalities invoke essentially no physics: no equations of motion, no energy conditions, and no boundary conditions. Even a Riemannian manifold is not required; the RT inequalities are in some sense really statements about weighted graphs \cite{Bao:2015bfa}. This is not true in the HRT setting; indeed, even basic inequalities such as strong subadditivity (SSA) are general relativity theorems, and their proofs make use of some basic physical requirements obeyed by the bulk spacetime, such as global hyperbolicity, the Einstein equation, the null energy condition, and AdS boundary conditions.\footnote{In proposed generalizations of the HRT formula outside of the classical, asymptotically AdS setting, such as \cite{Akers:2019lzs,Apolo:2020bld,Bousso:2023sya}, some of these specific conditions may be relaxed or altered, but the point remains that some dynamical information is required for proving SSA and other inequalities.} If the higher inequalities are valid, their proofs will also be general relativity theorems, which will tell us something new about holographic spacetimes. On the quantum information side, it is important to know whether the inequalities represent constraints on the entanglement structure of general holographic states, or merely static ones.

The validity of RT inequalities for HRT entropies has been shown in some special cases. In \cite{Wall:2012uf}, the maximin formulation of HRT was used to show that HRT obeys SSA and MMI. Unfortunately,  the above strategy cannot be used to prove the higher entropy inequalities, as presaged in \cite{Rota:2017ubr}.\footnote{Proving an inequality LHS $\ge$ RHS by maximin requires the RHS to be crossing-free. (Two boundary regions are said to cross if they partially overlap and don't cover the boundary; for example, if the boundary if divided into regions $A,B,C,D$, then $AB$ and $BC$ cross.) This holds for SSA and MMI. However, we have verified that all known higher inequalities involve crossing regions on the RHS.}\footnote{Very recently, the paper \cite{Bousso:2024ysg} claimed a proof of MMI for HRT using a new method, and suggested that this method may be applicable to higher inequalities.} In \cite{Bao:2018wwd}, it was shown that, for large regions and at late times in a thermalization process, where the membrane theory of entanglement spreading \cite{Jonay:2018yei} applies, any RT inequality is valid.

Further progress has been made in $2+1$ dimensional bulk setting, where an HRT surface is a union of spacelike geodesics. A few of the higher RT inequalities were tested numerically in AdS-Vaidya spacetimes, and found to be obeyed \cite{Erdmenger:2017gdk,Caginalp:2019mgu}. A proof that all RT inequalities are obeyed in $2+1$ dimensions was given by the authors of \cite{Czech:2019lps}. The proof proceeded in two steps. First, assuming that the bulk is simply connected, it was shown that, in a given configuration, any given RT inequality can be rewritten as a sum of instances of SSA, which is known to be obeyed by HRT \cite{Wall:2012uf}. The second step extends the proof to the case where the bulk is multiply connected by applying the argument in the first step to its universal covering space. Unfortunately, the second step is not correct: since the covering space has an infinite number of copies of the original spacetime, an infinite number of surfaces enter on both sides of the inequality, rendering it meaningless. 

Hence the validity of the RT inequalities for multiply-connected $2+1$ dimensional bulk spacetimes remains an open question. In this paper, we find evidence supporting their validity, in two ways:
\begin{itemize}
\item In section \ref{sec:numerical}, working with vacuum solutions, which are quotients of AdS$_3$ and where the lengths of spacelike geodesics can therefore be computed algebraically, we test thousands of inequalities (including all known inequalities up to 6 parties and several beyond that) on millions of randomly chosen configurations.
Our search covers three topological classes of solutions: $(n,g)=(2,0)$, $(3,0)$, and $(1,1)$, where $n$ is the number of boundaries and $g$ the genus. We find no counterexamples. We also report on the statistical distribution of the quantity $\Delta=\text{LHS}-\text{RHS}$ for several inequalities.\footnote{In appendix \ref{sec:quotients}, we give an overview of the computation of entanglement entropies in AdS$_3$ quotients. The \emph{Mathematica} code we developed for computing HRT entropies in AdS$_3$ quotients and testing inequalities is available as an ancillary file to this paper on the arXiv and at \href{https://sites.google.com/view/matthew-headrick/mathematica}{https://sites.google.com/view/matthew-headrick/mathematica}; see the notebook \tt{TestingHEIs.nb}.}
\item In section \ref{sec:argument}, building on the proof of \cite{Czech:2019lps} for simply-connected spacetimes, we give an argument that all RT inequalities are valid when the bulk has $\pi_1=\mathbb{Z}$, i.e.\ for $(2,0)$ wormholes.\footnote{In appendix \ref{sec:multipleHRT}, we discuss situations in which a boundary region admits multiple HRT surfaces, and prove a number of properties of such surfaces. As a corollary, we show that, if a boundary region is invariant under a bulk isometry, then it admits an invariant HRT surface, a result we use in section \ref{sec:argument}.} Unfortunately, as we explain, the argument does not immediately generalize to more complicated topologies.
\end{itemize}
We believe that these results, combined with the proof of \cite{Czech:2019lps} for simply connected spacetimes, constitute overwhelming evidence that all RT inequalities are obeyed by the HRT formula in $2+1$ dimensions.

Based on this conclusion, together with the earlier dimension-independent results of \cite{Wall:2012uf,Bao:2018wwd}, it seems very likely that the RT inequalities are obeyed by HRT in all dimensions. This claim is further substantiated by the perspective that the holographic entropy cone can be constructed using formal relations among ``proto-entropies'' \cite{Hubeny:2018trv,Hubeny:2018ijt} captured by the so-called pattern of marginal independence \cite{Hernandez-Cuenca:2022pst}, as these constructs have an inherently discrete structure which is therefore unlikely to be modified by time dependence. Unfortunately, a general proof that the HRT cone coincides with the RT cone is still lacking. We are fortunate by now to have several different but equivalent ways of expressing the HRT formula; in addition to maximin, there are the minimax, U-thread, and V-thread formulas \cite{Headrick:2022nbe}. The challenge before us is to use one of the known formulations --- or a new one --- to prove that these entropies obey the RT inequalities.\footnote{
In particular in upcoming work \cite{GGHH}, we explore the ingredients which suffice to use the minimax prescription for such a proof, which is closely related to the existence of a holographic graph model for general time-dependent spacetimes. 
}

\section{Numerical tests}
\label{sec:numerical}

The goal of this section is to provide numerical evidence for the validity of RT inequalities for the HRT formula when the bulk is a quotient of AdS$_3$. We will test all 1,877 currently known inequalities for $\mathsf{N} \leq 6$ regions \cite{Hernandez-Cuenca:2023iqh, hecdata}, as well as the members up to $\mathsf{k}=20$ of the dihedral family 
\be\label{eq:dihedral}
\sum_{i = 1}^{2\mathsf{k}+1} S(A_i \cdots A_{i+\mathsf{k}}) \geq \sum_{i = 1}^{2\mathsf{k}+1} S(A_i \cdots A_{i+\mathsf{k}-1}) + S(A_1\cdots A_{2\mathsf{k}+1})
\ee
(where $i$ takes values mod $2\mathsf{k}+1$), of which $\mathsf{k}=1$ is MMI and $\mathsf{k}=2$ is the 5-party inequality shown in \eqref{eq:5-party}. In what follows we set $4G_{\mathsf{N}} = 1$, so that all entropies are lengths.

As mentioned above, the RT inequalities have been proved in \cite{Czech:2019lps} when the bulk is simply connected; here, we are interested in testing them for multiply-connected bulks. In the context of pure 3d gravity with negative cosmological constant, these geometries are quotients of AdS$_3$. The resulting spacetimes have spatial slices that are wormholes with $n$ boundaries and genus $g$, or $(n,g)$-wormholes. A general framework to compute entanglement entropies for spacetimes of this kind was developed in \cite{Maxfield:2014kra}, which we review in appendix \ref{sec:quotients}. The important point for our purposes is that the calculation is purely algebraic, and as such does not require solving any differential equations, making it extremely fast. In this section, we make use of this speed to test thousands of distinct inequalities in millions of randomly chosen configurations of the bulk and of the boundary regions. These configurations fall into three topological classes: (i)   rotating BTZ, (ii) a rotating 3-boundary wormhole, and (iii) a rotating (1,1)-wormhole (i.e.\ a geometry with one asymptotic boundary and a genus-1 surface behind the horizon). We find no counterexamples. Along the way, we gather information about the kinds of configurations that saturate or nearly saturate certain inequalities, which may be of use for future studies of holographic entanglement.

As a control, and a check on our code, we also tried tweaking the inequalities, by dropping a term or changing a coefficient slightly, thereby rendering it false (for RT, and therefore for HRT as well). Indeed, in each case a counterexample was quickly found, usually within a few dozen random configurations. ``Counterexamples'' to RT inequalities are also quickly found with any change to the protocol for computing HRT entropies, for example by removing one of the allowed phases in the computation of the minimal geodesic length. We conclude contrapositively that \emph{not} finding a counterexample among millions of random configurations constitutes strong evidence for the validity of a given inequality, as well as for our code.

\subsection{(2,0)-wormholes (BTZ)}
\label{sec:BTZ}

The metric of the rotating BTZ solution reads
\be
\dd s^2 = -f(r) \dd t^2 + \frac{\dd r^2}{f(r)} +  r^2 \left(\dd \phi - \frac{r_{-}r_{+
}}{r^2}\dd t\right)^2, 
\ee
where $r_+ \geq r_- \geq 0$ are the outer and inner horizons, the coordinate $\phi$ is compact with period $\phi \sim \phi + 2\pi$ and $f(r)$ is the blackening factor
\be
f(r) = \frac{(r^2-r_{-}^2)(r^2-r_{+}^2)}{r^2}.
\ee
The two limiting cases, $r_{-} = 0$ and $r_{-} = r_{+}$, correspond to the static and extremal BTZ black holes respectively. The black hole has temperature
\be\label{eq:temperature}
T_{\mathsf{H}} = \frac{r_{+}^2 - r_{-}^2}{2\pi r_{+}}
\ee
and entropy
\be\label{eq:BTZentropy}
S_{\mathsf{BH}} = 2\pi r_{+} = 2\pi(T_{\mathsf{L}} + T_{\mathsf{R}}).
\ee
In the last equality, we have used the left- and right-moving temperatures\footnote{We absorb a non-standard factor of $\pi$ in the definition of the two temperatures for cleanliness of entropy formulas.}
\be
T_{\mathsf{L}} =  \frac{\pi T_{\mathsf{H}}}{1-\Omega} = \frac{r_+ + r_-}{2}\,, \qquad T_{\mathsf{R}} =  \frac{\pi T_{\mathsf{H}}}{1+\Omega} = \frac{r_+ - r_-}{2}\,,
\ee
with $\Omega = r_{-}/r_{+}$ the angular velocity of the black hole.

\begin{figure}
    \centering
    \includegraphics[width=.7\textwidth]{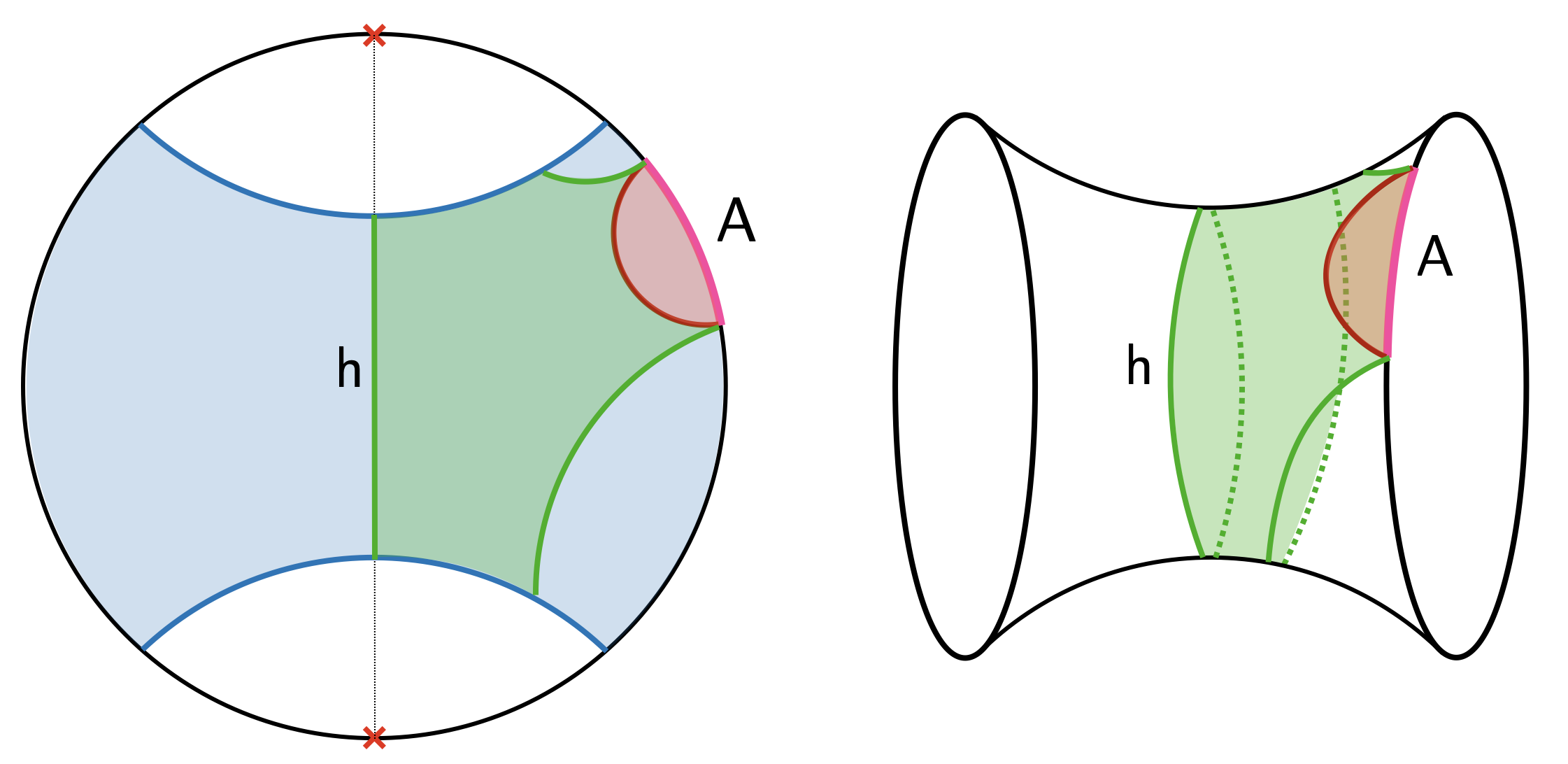}
    \caption{\textbf{[Left]} Spatial slice of the covering space, with the fundamental region of BTZ shaded in blue. The blue geodesics are to be identified under the action of the generator $g$ of the quotient group $\Gamma = \langle g \rangle$. We show a boundary region $\mathsf{A}$ in pink and the two phases of the entanglement entropy in red and green with the two homology regions shaded by the respective color. The horizon $\mathsf{h}$ is contained inside the (thin black) geodesic connecting the two fixed points of $g$, shown by the two red crosses on the disk. \textbf{[Right]} Spatial slice of the quotient geometry for BTZ, with the two phases shown.}
    \label{fig:BTZ}
\end{figure}

\subsubsection{AdS$_3$ quotient}

The BTZ black hole is the simplest quotient of AdS$_3$, corresponding to the case where the subgroup $\Gamma$ is generated by a single element, $\Gamma = \langle g \rangle$. The group $\Gamma$ will be a free group of rank 1, i.e. $\Gamma \simeq \mathbb{Z}$. The static solution corresponds to $\Gamma$ being diagonal, where the left and right elements $(g_L, g_R)$ are equal. Here we are interested in the more general solution and so  will not assume that $\Gamma$ is diagonal. To obtain the generator $g = (g_L, g_R)$ we can exponentiate a spacelike vector of our choice. As explained in appendix \ref{sec:quotients}, one simple choice is to set the $z$ and $t$ components from \eqref{eq:lie-algebra} to zero, obtaining
\be
\xi_{\mathsf{L},\mathsf{R}} = T_{\mathsf{L},\mathsf{R}}\begin{pmatrix}
    0 & 1 \\
    1 & 0
\end{pmatrix}, \qquad g_{\mathsf{L},\mathsf{R}} = e^{2\pi \xi_{\mathsf{L},\mathsf{R}}}=\begin{pmatrix}
    \cosh{2\pi T_{\mathsf{L},\mathsf{R}}}& \sinh{2\pi T_{\mathsf{L},\mathsf{R}}} \\
    \sinh{2\pi T_{\mathsf{L},\mathsf{R}}} & \cosh{2\pi T_{\mathsf{L},\mathsf{R}}}
\end{pmatrix},
\ee
where the $\mathsf{L},\mathsf{R}$ subscripts differentiate between the left and right action. This choice corresponds to a fundamental domain oriented so that the left and right boundaries are centered at $\phi = 0$ and $\phi = \pi$ as shown on the left side of figure \ref{fig:BTZ}. When the black hole is non-rotating one has $r_{-} = 0$, which implies $T_\mathsf{L} = T_\mathsf{R}$ and $g_L = g_R$ as expected from $\Gamma$ being diagonal. The fixed points of the generator $g$, being the eigenvectors, are the points $(1,1)$ and $(1,-1)$ which are located at $\phi = \pi/2$ and $\phi = 3\pi/2$ respectively. The geodesic connecting them contains the horizon of the black hole, as shown in figure \ref{fig:BTZ}. The element $g$ acts by moving points one full period around; therefore, $\xi_{\mathsf{L},\mathsf{R}}$ is the Killing vector that enacts spatial and time translation on the two asymptotic boundaries (though for time translations it acts oppositely on the left and on the right).

\subsubsection{Entanglement entropy}

We can now compute the entanglement entropy of a boundary interval $\mathsf{A}$ in the quotient geometry by using the formula \eqref{eq:lengthQuotient}. The intuition behind the formula is the following. From two endpoints in the quotient space, one gets infinitely many pairs of endpoints in the covering space. These copies of points are related by full $2\pi$ rotations on the boundary, which are enacted by the generator $g$ of $\Gamma$ and its inverse $g^{-1}$. As explained in appendix \ref{sec:quotients}, there exists a geodesic $\mathsf{g}_{\gamma}$ for each fixed endpoint homotopy class $\gamma \in \Gamma$, so there will be a geodesic anchored at $\partial \mathsf{A}$ for every element $\gamma \in \Gamma$.

Since $\Gamma = \langle g \rangle$ is generated by a single element, the most general geodesic will be in the homotopy class of some power of the generator, i.e. $\gamma = g^n$. In practice, the geodesic length of a geodesic $\mathsf{g}_\gamma$ can be computed by choosing two representative endpoints $\mathbf{p}$ and $\mathbf{q}$ (the simplest choice is to fix them to be on the boundary of the fundamental domain in the covering space), then translating the second endpoint $\mathbf{q}$ to the point $\gamma \mathbf{q}$ for the desired $\gamma\in\Gamma$. The length of the geodesic $\mathsf{g}_\gamma$ can then be easily computed in the covering space as the length of the geodesic connecting the endpoints $\mathbf{p}$ and $\gamma\mathbf{q}$, using \eqref{eq:lengthQuotient}. For example, choosing $\gamma = g^{2}$ corresponds (in the quotient space) to a geodesics that travels two full rotations around the wormhole before ending on the second endpoint and (in the covering space) to a homotopically trivial geodesic whose second endpoint has been translated down two copies. Hence, it's clear that most of these geodesics will not dominate in the computation for the entanglement entropy, since for higher powers of $g$ the geodesic $\mathsf{g}$ wraps around the wormhole more than once.

Thus, there will only be two competing geodesics: one in the homotopy class of the identity and one in the homotopy class of $g^{-1}$ (the geodesic in the homotopy class of $g$ does not work because it wraps around once but with a self intersection). When the region $\mathsf{A}$ has spatial and temporal openings $\Delta x$ and $\Delta t$, the first geodesic has length
\be
\ell(\Delta x,\Delta t) =\ln \sinh(T_{\mathsf{R}} (\Delta x- \Delta t))+\ln \sinh(T_{\mathsf{L}}(\Delta x+\Delta t))
\ee
and it is homotopic to the boundary interval, and thus trivially in the correct homology class. The second geodesic has length
\be
\ell(2\pi-\Delta x, -\Delta t)
\ee
and wraps around the horizon (in other words, it is the geodesic homotopic to the complement region $\mathsf{A}^c$). Following appendix \ref{sec:quotients}, to satisfy the homology constraint, this last geodesic must be augmented with the closed geodesic in the conjugacy class of $g$, which is just the horizon, with length $\mathsf{h}$ given by \eqref{eq:BTZentropy}. So the entanglement entropy for a single interval will be the minimum of these two, i.e.
\be
S(\mathsf{A}) = \min\left(\ell(\Delta x,\Delta t), \,\ell(2\pi-\Delta x, -\Delta t) + \mathsf{h}\right).
\ee

\begin{figure}
    \centering
    \includegraphics[width=0.7\textwidth]{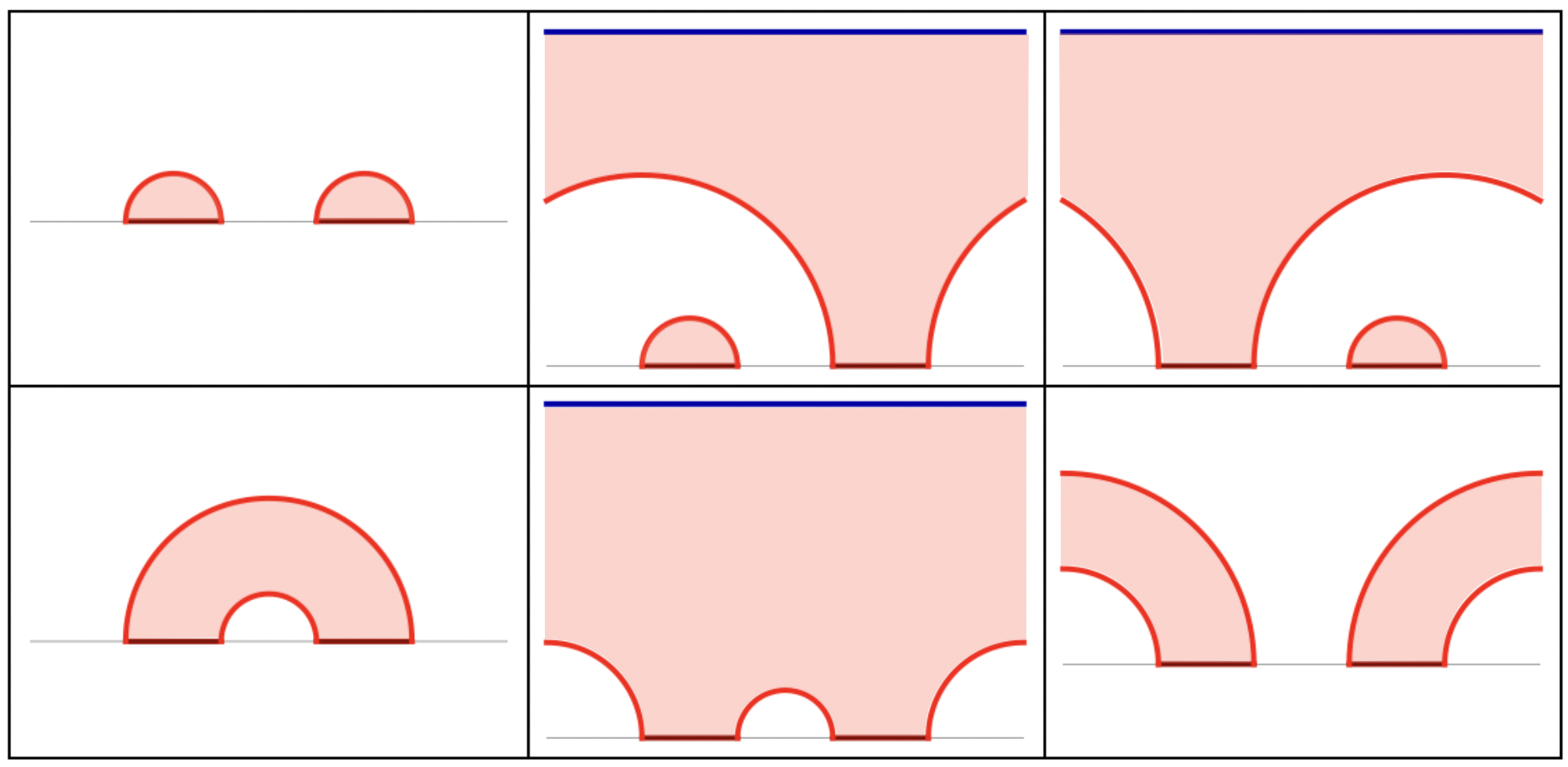}
    \caption{The six phases for the entanglement entropy $S(\mathsf{A}_1 \mathsf{A}_2)$ for two disconnected boundary intervals $\mathsf{A}_1$ and $\mathsf{A}_2$ in the BTZ geometry. We draw in red boundary anchored geodesics and in blue the horizon (if needed for the homology condition). The homology region is shaded in red. The left and ride edges of each panel are identified.}
    \label{fig:phases}
\end{figure}

We show in figure \ref{fig:BTZ} the two phases in the fundamental domain and the quotient. We will also be interested in computing the entanglement entropy for two or more disconnected boundary regions. By the same logic as above, for any two points on the boundary the only two competing geodesics will be the ones discussed above. Then, the entanglement entropy for more than one boundary region is found by listing all the possible ways of connecting the endpoints with two types of geodesics, taking into account the homology constraint. We do this combinatorial computation in \emph{Mathematica} and find that for $\mathsf{n}$ disconnected boundary intervals there are 
\be
\binom{2\mathsf{n}}{\mathsf{n}} = \frac{(2\mathsf{n})!}{(\mathsf{n}!)^2}
\ee
possible phases.\footnote{Actually, for $\mathsf{n} = 2$ where the boundary interval has endpoints $\mathbf{a}_1,\mathbf{a}_2,\mathbf{a}_3,\mathbf{a}_4$ there is one more phase where points $\mathbf{a}_1,\mathbf{a}_4$ are connected with a trivial geodesic and points $\mathbf{a}_2,\mathbf{a}_3$ are connected with a geodesic that wraps around the horizon (augmented with the horizon). While this phase has a self intersection it is still allowed topologically. However, it never dominates in the computation of entanglement entropy. Including it does not hurt, but we decided to list in figure \ref{fig:phases} only the phases that dominate. An analogous story holds when computing entropies for more than two disconnected regions.} For example, for $\mathsf{n} = 2$, there are six allowed phases, as shown in figure \ref{fig:phases}.

\subsubsection{Numerics}

To test the inequalities \eqref{eq:dihedral} numerically we need to first define the set of boundary regions. We want the boundary regions not to lie on a constant time slice in order to faithfully test the inequalities in the HRT regime. In particular, we will consider the following three types of configurations: 
\begin{enumerate}
    \item[a.] All the regions $\mathsf{A}_1,\mathsf{A}_2,\dots,\mathsf{A}_{2k+1}$ are single intervals covering one entire boundary (not in order) with no gaps. See figure \ref{fig:regions}.
    \item[b.] All the regions $\mathsf{A}_1,\mathsf{A}_2,\dots,\mathsf{A}_{2k+1}$ are single intervals covering one entire boundary in that order with a gap between the first and last one.
    \item[c.] The regions $\mathsf{A}_1,\mathsf{A}_2,\dots,\mathsf{A}_{2k}$ are single intervals covering one entire boundary in that order with a gap between the first and last one and region $\mathsf{A}_{2k+1}$ covers the entire second boundary.
\end{enumerate}
The rationale for fixing the order in case b is that generically this corresponds to the worst-case scenario, in the sense that the quantity $\Delta=\text{LHS}-\text{RHS}$ takes the smallest value when the HRT surfaces from the LHS and RHS are positioned maximally closely and deeply in the bulk, which for the form \eqref{eq:dihedral} happens when the regions are ordered.  Similar comment applies in case c, where we additionally induce the HRT surfaces to probe the topology more effectively. Therefore, while cases a, b, and c do not cover all the possibilities, we find it highly unlikely that a counterexample would be found elsewhere if none is found in these cases.

\begin{figure}
    \centering
    \includegraphics[width=\textwidth]{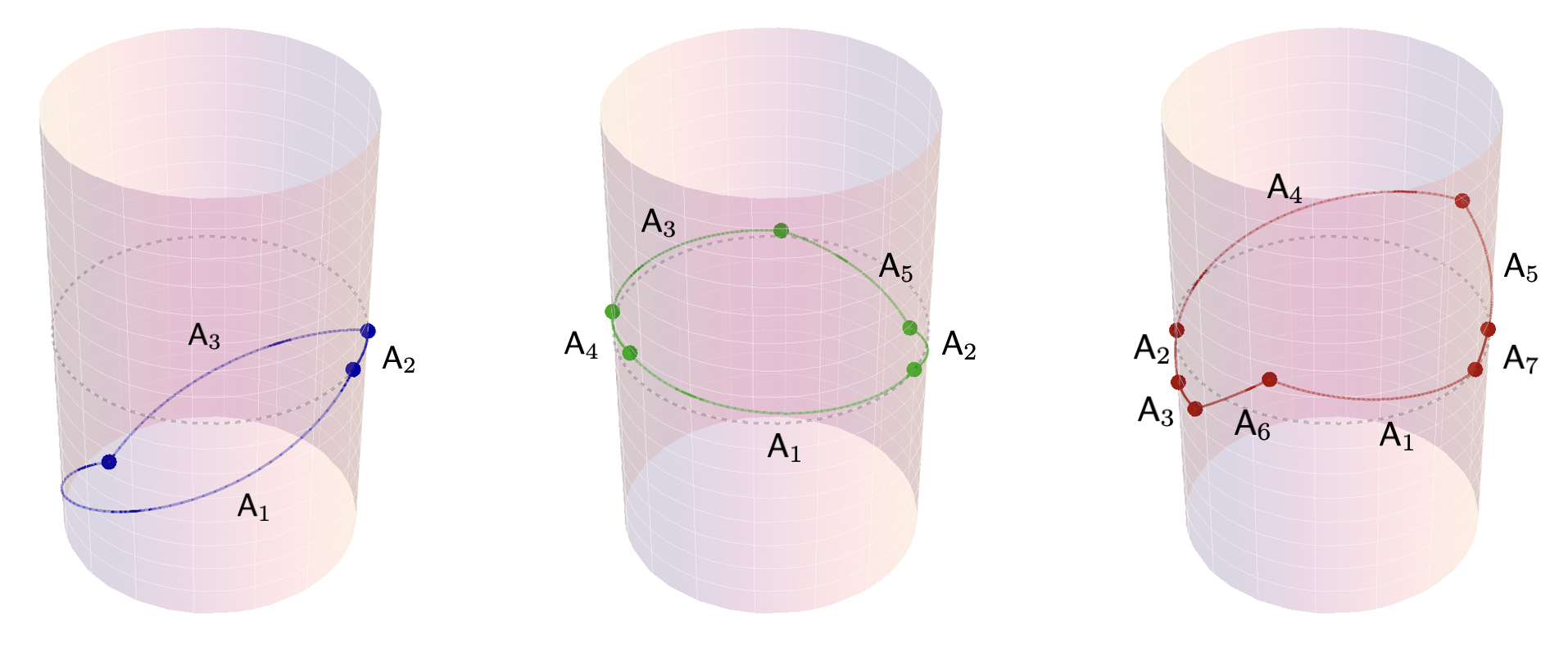}
    \caption{Examples of configurations of boundary regions for (i) $\mathsf{k} = 1$, (ii) $\mathsf{k} = 2$ and (iii) $\mathsf{k} = 3$. All three configurations are of Type a --- the regions cover one entire asymptotic boundary and the regions are not in order (except, of course, for $\mathsf{k} = 1$). The dashed gray line represents the $t = 0$ slice. All the regions lie on a boundary Cauchy slice.}
    \label{fig:regions}
\end{figure}

The code is structured as follows:
\begin{enumerate}
    \item To test the dihedral inequality \eqref{eq:dihedral} for a given value of $\mathsf{k}$, it outputs a random set of either $2\mathsf{k}+1$ (cases a and c) or $2\mathsf{k}+2$ (case b) spacelike-separated points, by first generating the $\phi$-values with a uniform probability distribution in the interval $(0,2\pi)$, then picking $t$-values by multiplying each coordinate by a random number in the interval $(-1,1)$ and finally checking that the resulting set of points lies on a boundary Cauchy slice (if not, it keeps repeating the same process until it does). These points will define the endpoints of the regions. For type-a regions, the code further selects a random element of the permutation group $S_{2\mathsf{k}+1}$ and permutes the labels of the regions so that they are not in order. See figure \ref{fig:regions} for an example. To test all the other primitive HEIs, the code is fed an entropy vector corresponding to a specific information quantity. It then generates a random set of 7 regions (6 regions + one gap) in the same manner as above, with the labels once again permuted.
    \item Once the boundary regions are defined, the code selects left and right moving temperatures $T_{\mathsf{L}}$ and $T_{\mathsf{R}}$  from a log-normal distribution with mean $\mu = 0$ and standard deviation $\sigma = 2$. We further truncate the distribution (looking at black holes with radius $\leq 10$) to avoid selecting black holes with too large of a horizon, as explained below. The random pair $\{T_{\mathsf{L}},T_{\mathsf{R}}\}$ uniquely defines the spacetime.
    \item All the required entropies for the terms entering the inequality are then computed, and the program stores the following array of information about the trial
    \be
    \{ \{T_{\mathsf{L}}, T_{\mathsf{R}}\}, \textsc{points}, \Delta\},
    \ee
    where $\textsc{points}$ is an array containing the coordinates of the endpoints defining the boundary regions and $\Delta = \text{LHS}-\text{RHS}$ is the value of the inequality. For type-a regions we also store the element of the permutation group used to shuffle the regions. 
    \item The code repeats steps (1)--(3) for $n$ independent trials.
\end{enumerate}

\begin{figure}
    \centering
    \includegraphics[width=\textwidth]{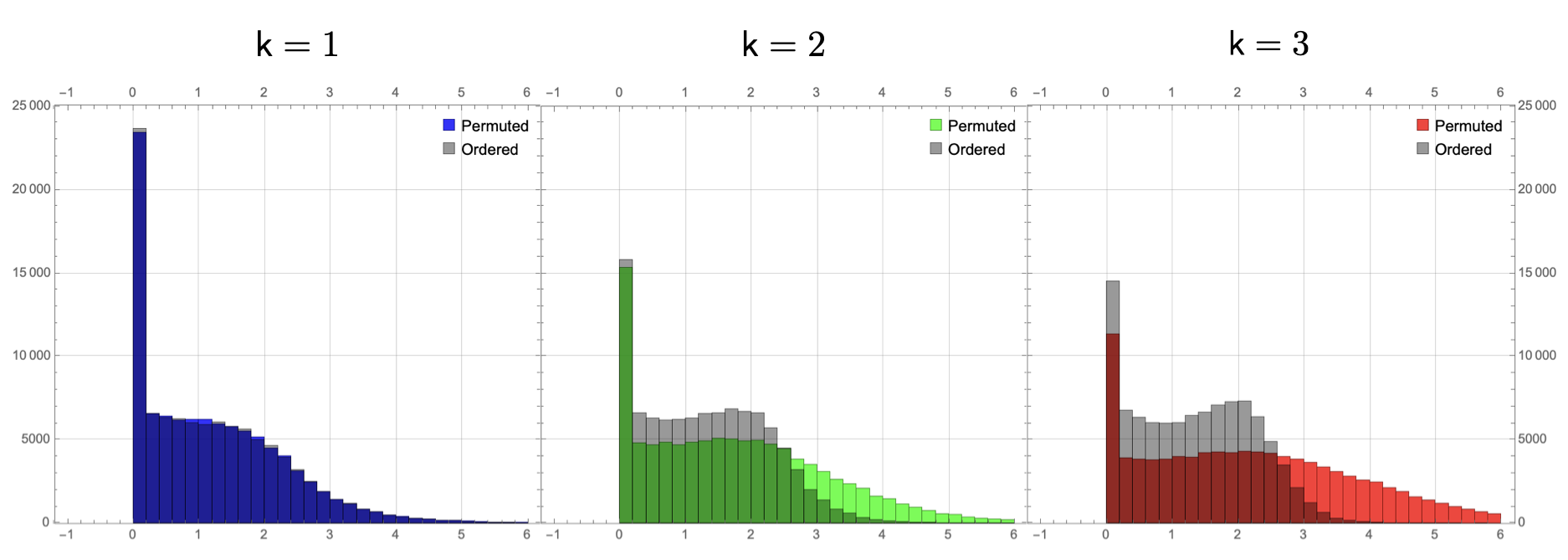}
    \caption{Distribution of $\Delta = \text{LHS}-\text{RHS}$ from $\mathsf{k} = 1$ to $\mathsf{k} = 3$. All boundary regions are of type a and the bulk is BTZ. In each plot we show the distribution for permuted (colored) and ordered (gray) regions. For $\mathsf{k} = 1$ there is, of course, no difference but for $\mathsf{k} = 2,3$ the ordered configuration has smaller mean.} 
    \label{fig:histo}
\end{figure}

We performed the following tests:
\begin{itemize}
  \item For dihedral inequalities:
  \begin{itemize}
    \item For type a regions, we collected $n = 10^5$ trials for all dihedral inequalities up to  $\mathsf{k} = 5$ (corresponding to $\mathsf{N}=11$ parties). For the first three dihedral inequalities we also ran trials without permuting the regions, and we show the results in figure \ref{fig:histo}. This confirms pictorially that the ordered case is statistically closer to saturation.
    \item For both type b and type c regions, we collected $n = 10^5$ trials for all dihedral inequalities up to $\mathsf{k} = 10$.
  \end{itemize}
  \item For the 6-party inequalities obtained in \cite{Hernandez-Cuenca:2023iqh}, we ran $n = 10^5$ trials for every inequality up to $Q_{10}$ (plus $Q_{22}$ and $Q_{39}$) and $n =  10^3$ trials for all the remaining 1,865 HEIs. We choose to show in figure \ref{fig:histo2} the inequalities $Q_{12},Q_{22}$ and $Q_{39}$ as they are good representatives for the different structure of the 6-party HEIs, as can be seen when written in a compact form using tripartite information and conditional tripartite information:
\begin{align}\label{eq:6-party_ex}
\begin{split}
Q_{12} &= -I_3(A:BE:CF)-I_3(B:C:D|A)
\\
Q_{22} &= -I_3(A:BD:CE)-I_3(AC:B:EF)-I_3(C:D:F|AB)
\\
Q_{39} &= -I_3(A:CF:DE)-I_3(B:CE:DF)-I_3(C:D:E)-I_3(A:B:F|E)
\end{split}
\end{align}
  \end{itemize}

  \begin{figure}
    \centering
    \includegraphics[width=\textwidth]{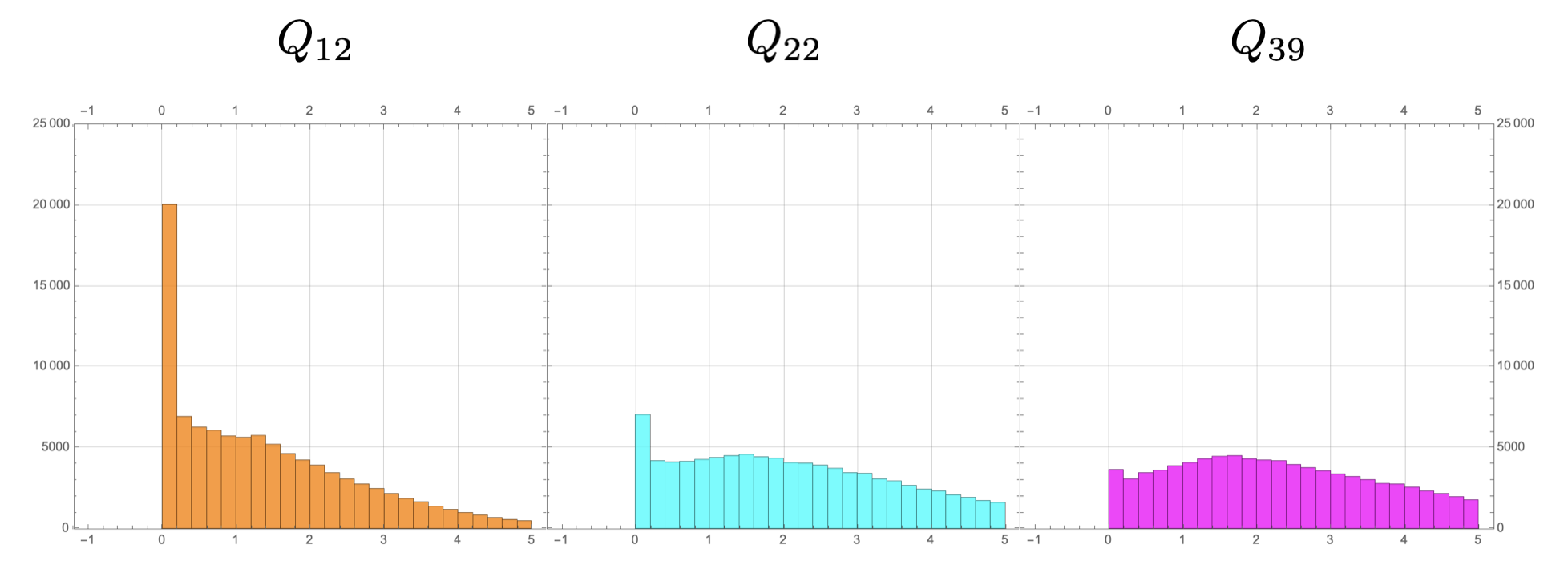}
\caption{Distribution of $\Delta = \text{LHS}-\text{RHS}$ for the inequalities $Q_{12}, Q_{22}$ and $Q_{39}$ written explicitly in \eqref{eq:6-party_ex}. All boundary regions are of Type b and the bulk is a rotating BTZ.} 
    \label{fig:histo2}
\end{figure}

No counterexample was found. Furthermore, as can be seen from figures \ref{fig:histo} and \ref{fig:histo2}, note that there is an accumulation of configurations that saturate or nearly saturate the inequalities. This suggests that a counterexample --- if it existed --- would have plausibly been found. 

We also note that, as remarked above, configurations where the subregions are in order on the boundary have smallest average $\Delta$. This is a particularly helpful fact since for unpermuted configurations all the entropies appearing in an inequality are of single connected intervals, making the numerics not only simpler but also much faster.

It is also interesting to study the nature of the saturating configurations, as they are the best starting point to find counterexamples. This amounts to looking at which phase dominates for each term entering the inequality. We do this for MMI ($\mathsf{k} = 1$) and the 5-party dihedral ($\mathsf{k} = 2$), and find that the configurations can be classified into two groups. The first and most common (accounting for about $99$\% of cases) happens when the LHS and the RHS surfaces are the same, hence cancelling out. If the boundary regions are ordered, this happens when a region $\mathsf{A}$ on the LHS and its complement $\mathsf{A}^c$ on the RHS share the same phase (modulo an extra horizon). If further the phases are such that one has the same number of horizons on the LHS and RHS, we have an exact cancellation. Otherwise the inequality will be an exact even multiple of the horizon's length, i.e.
\be
\Delta = 2m\mathsf{h}\,, \qquad m = 1,2\dots,\mathsf{k}\,.
\ee
We note that these type of saturating configurations are robust to perturbations of the geometry (e.g.\ if matter is added) and therefore suggest that our numerics should hold for an open neighborhood of spacetimes near these pure AdS$_3$ quotients. This intuition will later be confirmed by the argument in section \ref{sec:argument}, where we do not specialize to vacuum solutions.

The second and rarer class is when the black hole is very large. In this case, the second phase never dominates and all the geodesics are homotopic to the regions. All the LHS surfaces wrap around the horizon, compensating for the horizon term on the RHS. We note that this class of configurations does not saturate the inequality exactly, but it does so with an error that gets smaller as the horizon size increases. For this reason, we don't consider this second class of configurations to be of any physical importance and it is easily removed by adding a cutoff when sampling right and left moving temperatures. The fact that none of the saturating or near-saturating configurations lend themselves to over-saturation under small deformations (in particular by utilizing time dependence) gives further evidence that HRT cone
coincides with the RT cone.

\subsection{(3,0)-wormholes}
\label{sec:(3,0)}

The next spacetime we consider is the (rotating) 3-boundary wormhole, a geometry where the spatial slice has the topology of a pair of pants. The CFT interpretation of this state is of a pure entangled state of three circles.

\subsubsection{AdS$_3$ quotient}

The spatial geometry is conformally equivalent to a three-punctured sphere, whose fundamental group is the free group of rank 2. Therefore, the quotient $\Gamma$ of a (3,0)-wormhole has two generators, i.e.\ $\Gamma = \langle g_1, g_2 \rangle$.

\begin{figure}
    \centering
    \includegraphics[width=.45\textwidth]{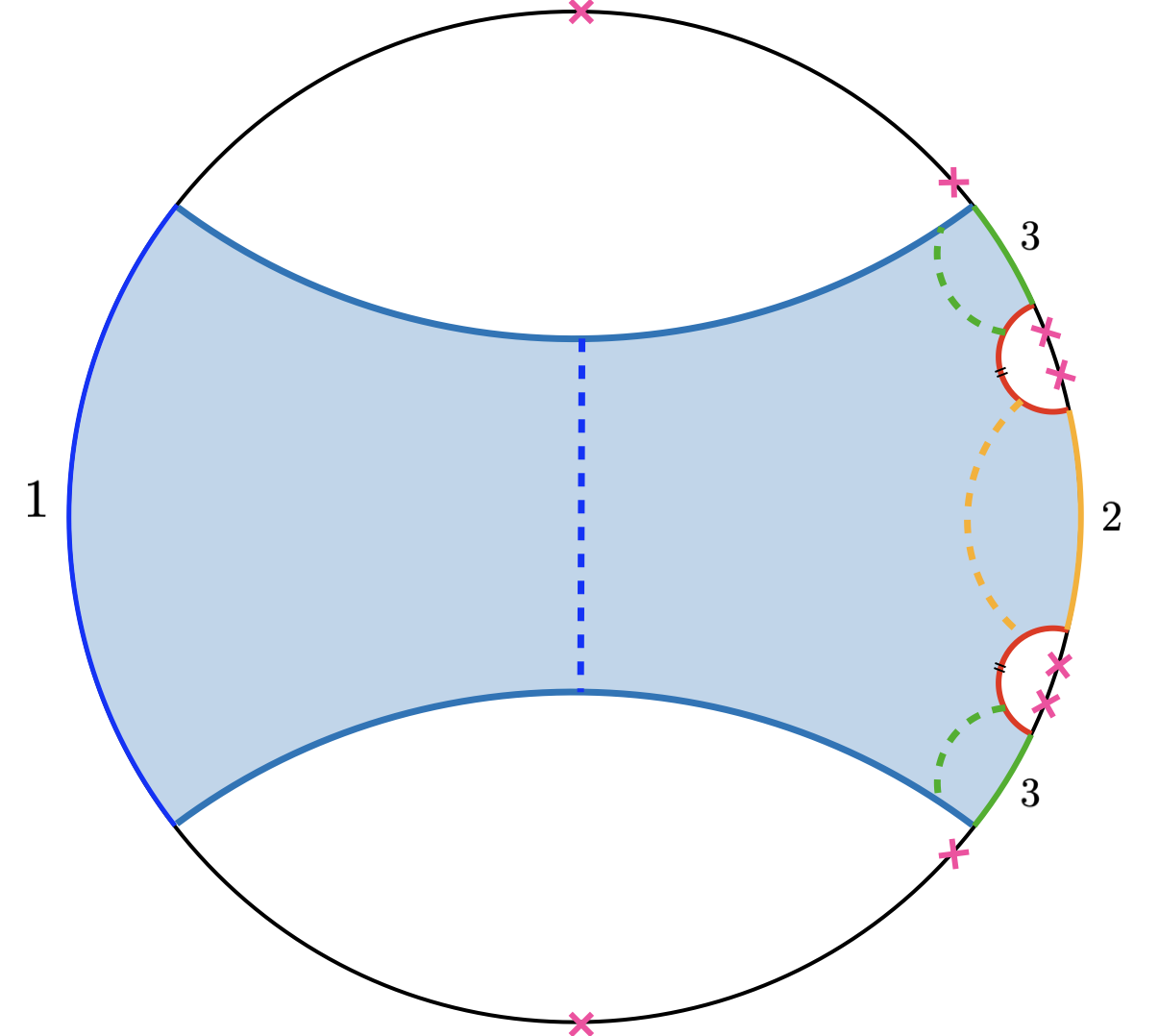}
    \caption{Fundamental domain of the 3-boundary wormhole. The pair of blue and red geodesics are to be identified under the action of the two generators $g_1$ and $g_2$. The three horizons are contained inside the geodesics connecting the pairs of fixed points (shown as pink crosses on the covering space boundary) of $g_1$, $g_2$ and $g_1 g_2^{-1}$.}
    \label{fig:3bdy-covering}
\end{figure}

The static case was carefully studied in \cite{Maxfield:2014kra}. The generalization to the rotating case is relatively straightforward and we will present it here. To obtain a spatial slice with the topology of a pair of pants, it is simplest to start with a fundamental region as shown in figure \ref{fig:3bdy-covering} where we glue two pairs of geodesics. The action of the two generators $g_1$ and $g_2$ identifies the two pairs. For the choice of fundamental domain we have made, one of the spacelike vectors can be chosen to be the same as the one for BTZ. The other can be obtained by boosting this vector in the $z$ direction (which we recall was along the $x$ axis; see \eqref{eq:lie-algebra}) with an angle $\alpha$, i.e.
\be
\xi_{1}^{\mathsf{L},\mathsf{R}} = \ell_{1}^{\mathsf{L},\mathsf{R}}\begin{pmatrix}
0 & 1\\
1 & 0
\end{pmatrix}, \qquad 
\xi_{2}^{\mathsf{L},\mathsf{R}} = \ell_{2}^{\mathsf{L},\mathsf{R}}\begin{pmatrix}
0 & e^{\alpha_{\mathsf{L},\mathsf{R}}}\\
e^{-\alpha_{\mathsf{L},\mathsf{R}}}& 0
\end{pmatrix}.
\ee
Since the wormhole is rotating, there are a total of six moduli that fully characterize the geometry: two pairs of three moduli for the left and right components. The four parameters $\ell_1^{\mathsf{L},\mathsf{R}}$ and $\ell_2^{\mathsf{L},\mathsf{R}}$ are analogous to the left and right moving temperatures in the rotating BTZ example, and are related to the lengths of the horizons for the first and second asymptotic regions. As with any quotient, each asymptotic region is BTZ-like and will have a generator of translations associated with it. The elements associated to the first and second horizons are $g_1$ and $g_2$ (as can be checked from their fixed points), and so the $\xi_1$ and $\xi_2$ generate translations on the first and second boundaries. Indeed, the length of the two horizons can be computed to be
\be
\mathsf{h}_1 = \cosh^{-1}\left(\frac{\Tr g_1^{\mathsf{L}}}{2}\right) + \cosh^{-1}\left(\frac{\Tr g_1^{\mathsf{R}}}{2}\right)  = \ell_1^\mathsf{L} + \ell_1^\mathsf{R}
\ee
for the first horizon, and 
\be
\mathsf{h}_2 = \cosh^{-1}\left(\frac{\Tr g_2^{\mathsf{L}}}{2}\right) + \cosh^{-1}\left(\frac{\Tr g_2^{\mathsf{R}}}{2}\right)  = \ell_2^\mathsf{L} + \ell_2^\mathsf{R}
\ee
for the second horizon. The corresponding element for the third horizon is $g_1 g_2^{-1}$, and the corresponding Killing vector $\xi_3$ can be found from solving the matrix equation $e^{2\pi \xi_3} = g_1 g_2^{-1}$. With the current choice of parameters the information about the third horizon is hidden inside the boost angles $\alpha_{\mathsf{L},\mathsf{R}}$. To obtain the length of the third horizon we can use the formula for the length of closed geodesics, which gives us
\begin{align}\label{eq:3rdhorizon}
\mathsf{h}_3 &= \cosh^{-1}\left[\frac{\Tr(g_1^{\mathsf{L}}(g_2^{\mathsf{L}})^{-1})}{2}\right] + \cosh^{-1}\left[\frac{\Tr(g_1^{\mathsf{R}}(g_2^{\mathsf{R}})^{-1})}{2}\right]\\
&= \sum_{\mathsf{I}\, = \,\{\mathsf{L},\mathsf{R}\}}\cosh^{-1}\left(\cosh{\alpha_{\mathsf{I}}}\sinh\ell_1^{\mathsf{I}}\sinh\ell_2^{\mathsf{I}}-\cosh\ell_1^{\mathsf{I}}\cosh\ell_2^{\mathsf{I}}\right).
\end{align}
It is important to note that the element $g_1 g_2^{-1}$ will be a valid hyperbolic element only when its trace is larger than two. This imposes a lower bound for the values of the angles $\alpha_\mathsf{L}$ and $\alpha_\mathsf{R}$, or equivalently, it imposes that $\mathsf{h}_3 > 0$. Because of this, we will from now on define the geometry by the six positive left/right lengths $(\ell_1^\mathsf{L}, \ell_1^\mathsf{R}, \ell_2^\mathsf{L}, \ell_2^\mathsf{R}, \ell_3^\mathsf{L}, \ell_3^\mathsf{R})$, with $\ell_3^{\mathsf{L},\mathsf{R}}$ being the first and second terms in \eqref{eq:3rdhorizon}.

\subsubsection{Entanglement entropy}

\begin{figure}
    \centering
    \includegraphics[width=0.7\textwidth]{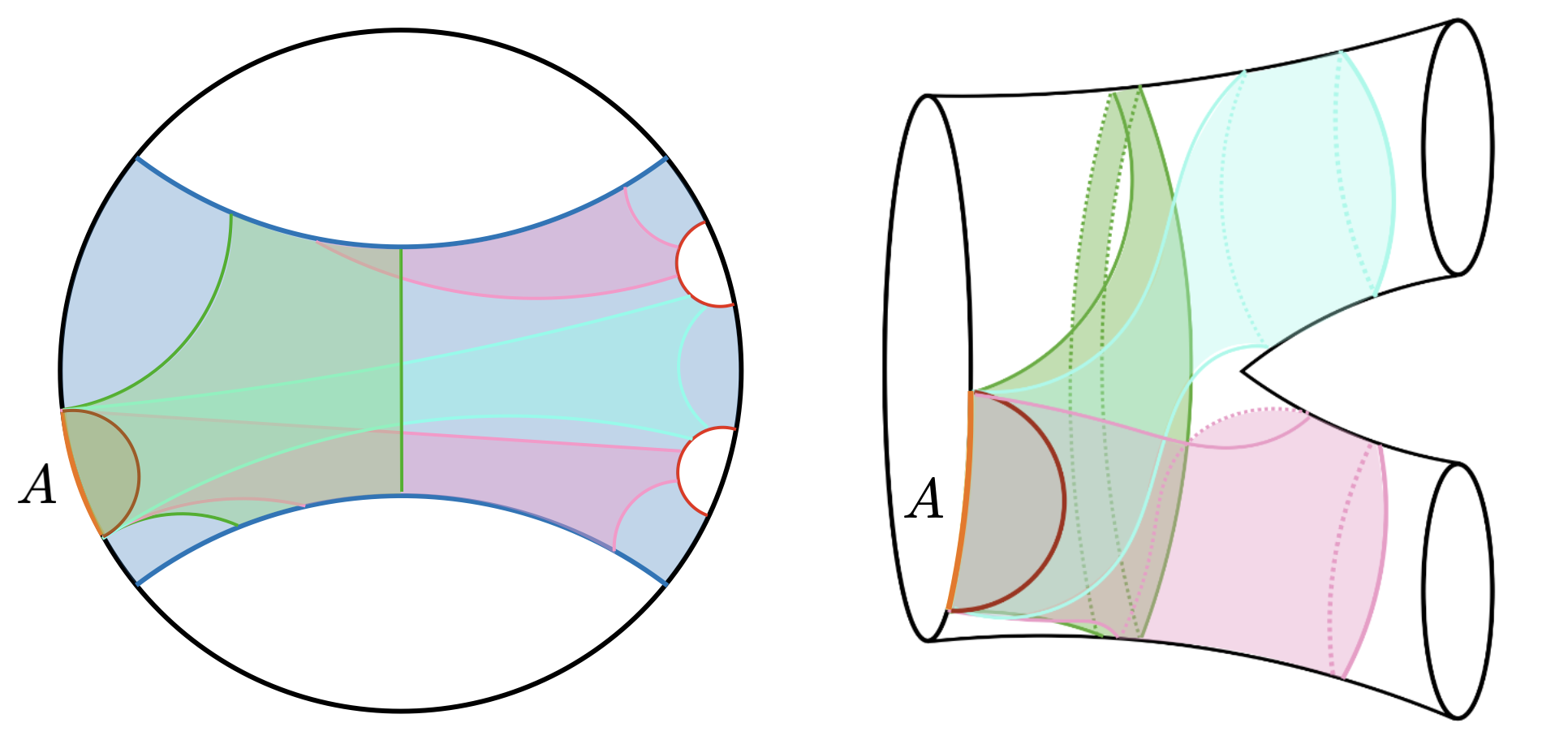}
    \caption{\textbf{[Left]} The four phases of entanglement entropy for the boundary region $\mathsf{A}$ in the fundamental domain, with the  homology regions for each phase shaded. In red we have the geodesic homotopic to the interval. In green, cyan and pink we show the geodesics in the $g_{1}^{-1}, g_2^{-1}$ and $g_2g_1^{-1}$ homotopy classes augmented with the first, second, and third horizons respectively. The green geodesic may be augmented with the union of the second and third horizons if their total length is smaller.  \textbf{[Right]} The same four phases on a spatial slice of the quotient.}
    \label{fig:3bdy-quotient}
\end{figure}

Consider a region $\mathsf{A}$ on the first boundary. Unlike the two phases for BTZ which depended only on the sizes $\Delta x$ and $\Delta t$ of the interval, for the 3-boundary wormhole the phases will depend on the position of the interval itself. Because of this, we must carefully keep track of the endpoints of the region $\mathsf{A}$, call them $\mathbf{a}_1$ and $\mathbf{a}_2$ with coordinates $(t_1,\phi_1)$ and $(t_2,\phi_2)$ on the quotient's boundary. To compute the geodesic length using \eqref{eq:lengthQuotient}, we must find the location of $\mathbf{a}_1$ and $\mathbf{a}_2$ on the fundamental domain in the covering space. To do that, we begin by defining a fixed representative point $\mathbf{p}$ for the first asymptotic region, the simplest choice being the midpoint $t = 0$ and $\phi = \pi$ on the fundamental domain. Then, we translate $\mathbf{p}$ in space and time using the appropriate Killing vectors to obtain $\mathbf{a}_1$ and $\mathbf{a}_2$, as follows
\be
\mathbf{a}_1 = e^{(\phi_1 + t_1)\, \xi_1^{\mathsf{L}} }\mathbf{p}\,e^{(\phi_1 - t_1)\, \xi_1^{\mathsf{R}}}, \quad \mathbf{a}_2 = e^{(\phi_2 + t_2)\, \xi_1^{\mathsf{L}} }\,\mathbf{p}\,e^{(\phi_2 - t_2)\, \xi_1^{\mathsf{R}}}
\ee
where, we recall, the sign difference for the time translation is there because time translations act oppositely for the left and right components. We can then use \eqref{eq:lengthQuotient} to compute the geodesic lengths $\ell(\mathbf{a}_1, \mathbf{a}_2, \gamma)$ in the fixed endpoint homotopy class of some $\gamma\in\Gamma$. As in the BTZ example, we need to consider only those that can dominate the computation for the entanglement entropy. It was shown in \cite{Maxfield:2014kra} that there are four such phases, and they are the geodesic in the homotopy classes of the identity, $g_1^{-1}$, $g_2^{-1}$ and $g_2g_1^{-1}$.  The first two are the same as the BTZ phases, with the only difference being that for the second phase one may satisfy the homology constraint by augmenting the geodesic with either the first horizon $\mathsf{h}_1$ or the union of the other two $\mathsf{h}_2 \cup \mathsf{h}_3$ if smaller. From an algebraic perspective including the closed geodesics $\mathsf{h}_2 \cup \mathsf{h}_3$ is allowed since they are in the conjugacy classes of $g_2$ and $g_1g_2^{-1}$, which, together with the boundary anchored geodesic in the homotopy class of $g_1^{-1}$ give the correct trivial homology (once abelianized). The more interesting phases are the third and fourth which are unique to the three boundary wormhole. They correspond to geodesics that wrap around the second and third leg, respectively. The homology condition is satisfied by augmenting them with the second and third horizons, respectively.  We show in figure \ref{fig:3bdy-quotient} the four phases both in the fundamental domain and in the quotient.

\subsubsection{Numerics}

For simplicity, we tested the dihedral family for the case where all the boundary regions $\mathsf{A}_1, \dots, \mathsf{A}_{2k+1}$ are single intervals (in that order) covering the first boundary with no gaps. This choice is supported by the evidence from the BTZ example, where we saw that permuting Type-a regions has the effect of increasing the average value of $\Delta = $ LHS$-$RHS. In this configuration every entropy appearing in any inequality is the entropy of a single connected interval.

While naively it would appear that such configuration does not faithfully test the inequalities for the 3-boundary wormhole (since all the regions are restricted to a single boundary which is asymptotically BTZ) it still offers a highly non-trivial check since, as established in the previous section, two of the four phases that compute the entropy of a single interval leave the exterior region of the boundary, therefore probing the interior geometry. In other words, the computation of the entropy knows about the full topology of the slice, i.e.\ the formulae depend on the six moduli of the pants.

\begin{figure}
    \centering
\includegraphics[width=\textwidth]{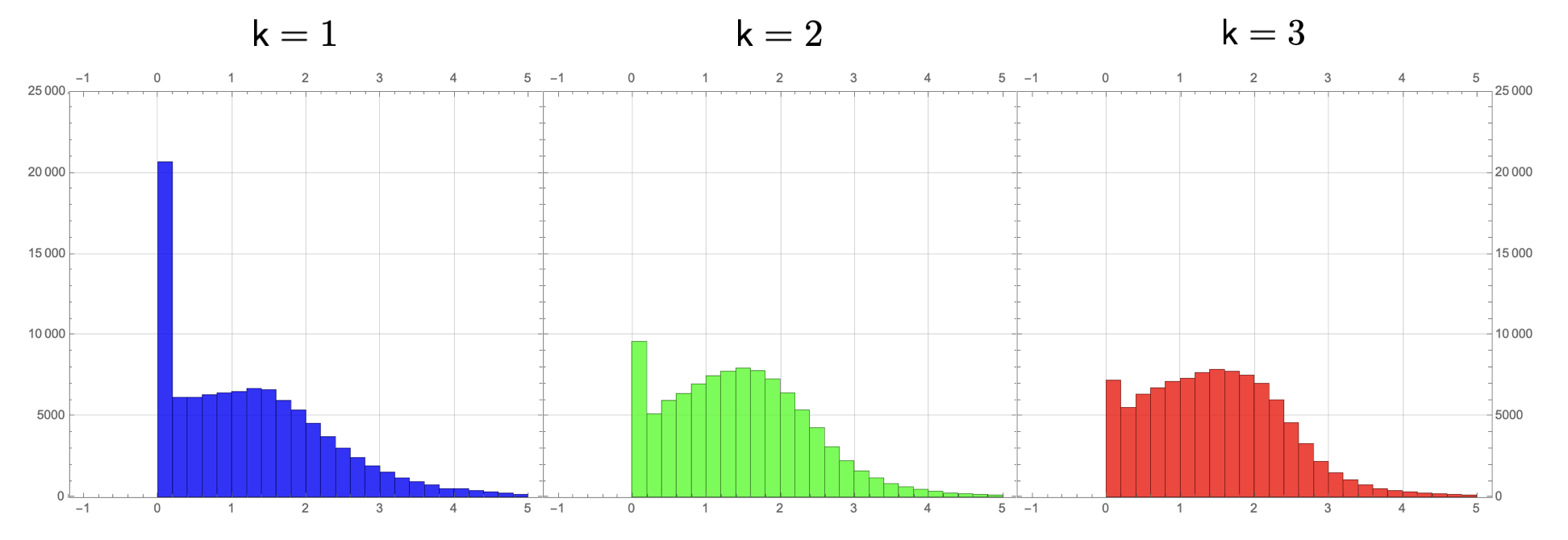}
    \caption{Distribution of $\Delta = \text{LHS}-\text{RHS}$ for the first three dihedral inequalities when the bulk is a rotating (3,0) wormhole. For each iteration in the code the six moduli $(\ell_1^\mathsf{L}, \ell_1^\mathsf{R}, \ell_2^\mathsf{L}, \ell_2^\mathsf{R}, \ell_3^\mathsf{L}, \ell_3^\mathsf{R})$ have been generated from a uniform random distribution in the range $(0,20)$.} 
    \label{fig:histo3bdy}
\end{figure}

Aside from the more complicated setup, the code is largely the same as the BTZ example. One difference worth mentioning is that since we are mostly interested in testing the inequalities in a regime where the phases unique to the (3,0)-geometry dominate, which happens when $\mathsf{h}_1\gg \mathsf{h}_{2,3}$ and $\mathsf{h}_1$ is larger than some order one number (see \cite{Maxfield:2014kra} for a detailed analysis), we  randomly generated each of the six moduli with a uniform probability distribution in the interval $(0,20)$. Similarly to the BTZ code, the program stores the following array of information about each trial
    \be
    \{ \{\ell_1^\mathsf{L}, \ell_1^\mathsf{R}, \ell_2^\mathsf{L}, \ell_2^\mathsf{R}, \ell_3^\mathsf{L}, \ell_3^\mathsf{R}\}, \textsc{points}, \Delta\}.
    \ee
and repeats the process for $n$ independent trials. We performed the following tests:
\begin{itemize}
    \item For the first three dihedral inequalities we collected $n = 10^5$ trials, and we show the results in figure \ref{fig:histo3bdy}.
    \item For each of the remaining dihedrals up to $\mathsf{k} = 20$ (corresponding to $\mathsf{N} = 41$ parties), we collected $n = 10^4$ trials.
\end{itemize} 
No counterexample was found.

\subsection{(1,1)-wormholes}
\label{sec:(1,1)}

We now move our attention to our third and final quotient: the (1,1)-wormhole, also known as the torus wormhole. This is a vacuum solution where the spatial slice has only one asymptotic boundary (the dual CFT state is pure) but with non-trivial topology behind the horizon (it has genus one).

\subsubsection{AdS$_3$ quotient}

The spatial geometry is conformally equivalent to a one-punctured torus, which shares the same fundamental group as a three-punctured sphere. Therefore, this geometry has the same quotient group $\Gamma$ as the (3,0)-wormhole: the free group of rank 2. So $\Gamma$ has two generators, i.e. $\Gamma = \langle g_1, g_2 \rangle$. The difference with the $(3,0)$ wormhole lies in how these generators act when gluing pairs of geodesics in the fundamental domain: in the (1,1)-wormhole opposite pairs of geodesics are identified, whereas in the (3,0)-wormhole the identification is performed on adjacent pairs.

In \cite{Maxfield:2014kra} a symmetric non-rotating version of this geometry was studied. There, the geometry was fully characterized by one modulus (the length of the horizon). However, in general the geometry is defined by 6 moduli (two pairs of three moduli for the left and right components), just like the (3,0)-wormhole. Below we will give a description of this geometry in generality but for our numerics we will stick to a rotating generalization of the symmetric solution studied in \cite{Maxfield:2014kra} (the same as the one appearing in the original paper \cite{Aminneborg:1998si}). 

\begin{figure}
    \centering
    \includegraphics[width=0.6\textwidth]{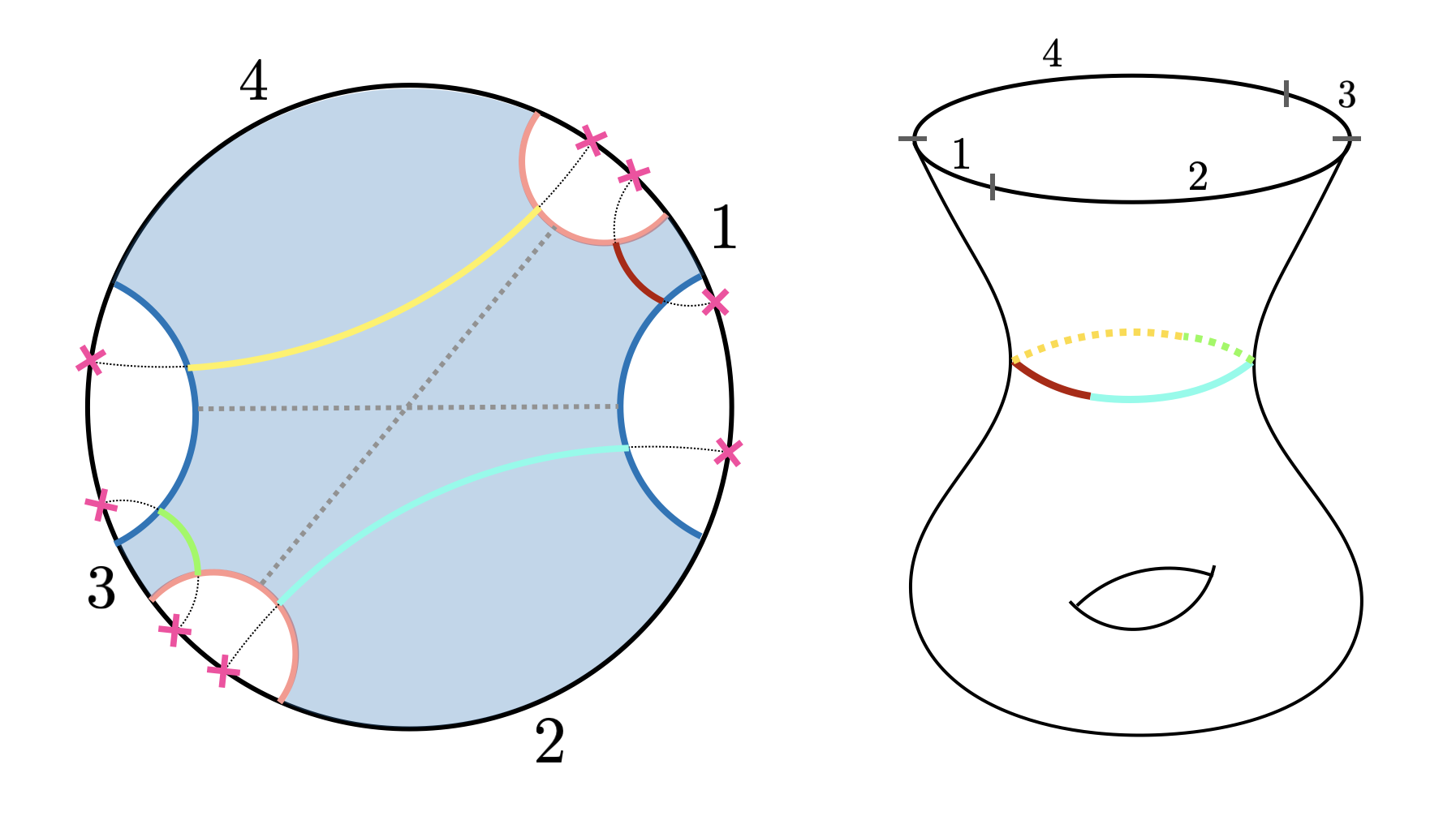}
    \caption{\textbf{[Left]} Fundamental domain for the torus wormhole. The two generators $g_1$ and $g_2$ enact the gluing of the blue and red pairs of geodesics respectively. Such identification will form a close cycle on the boundary, and thus the resulting quotient will have a single boundary and  a single horizon formed by taking the union of the four colored segments in the fundamental domain. The two dotted lines in gray are the closed geodesics associated to the two generators $g_1$ and $g_2$ which, in the quotient geometry, result in the two non-contractible loops for the torus behind the horizon. \textbf{[Right]} Resulting spatial slice of the quotient after we perform the identifications on the fundamental domain on the left.}
    \label{fig:torus-wormhole-quotient}
\end{figure}

Defining the fundamental domain as in figure \ref{fig:torus-wormhole-quotient}, one of the Lie algebra generators can be chosen to lie along the $z$-axis (identifying the pair of blue geodesics) and the other one by rotating the first with some angle $\alpha$ (identifying the pair of pink geodesics). Using \ref{eq:lie-algebra}, we have
\be
\xi_1^{\mathsf{L},\mathsf{R}} = \ell_1^{\mathsf{L},\mathsf{R}}\begin{pmatrix}
1 & 0\\
0 & -1
\end{pmatrix}\quad \text{and}\quad
\xi_2^{\mathsf{L},\mathsf{R}} = \ell_2^{\mathsf{L},\mathsf{R}}\begin{pmatrix}
\cos{\alpha_{\mathsf{L},\mathsf{R}}} & \sin{\alpha_{\mathsf{L},\mathsf{R}}}\\
\sin{\alpha_{\mathsf{L},\mathsf{R}}} & -\cos{\alpha_{\mathsf{L},\mathsf{R}}}
\end{pmatrix}.
\ee
One important comment to make right away is that neither $g_1$ nor $g_2$ are the elements in the conjugacy class corresponding to the closed geodesics that contain the horizon.  By following the identifications on the fundamental domain, one finds that the two generators are associated to the two non-contractible cycles of the torus behind the horizon, with lengths 
\be
\ell_1 = \cosh^{-1}\left(\frac{\Tr g_{1}^\mathsf{L}}{2}\right) + \cosh^{-1}\left(\frac{\Tr g_1^{\mathsf{R}}}{2}\right) = \ell_1^{\mathsf{L}} + \ell_1^{\mathsf{R}}
\ee
and
\be
\ell_2 = \cosh^{-1}\left(\frac{\Tr g_2^\mathsf{L}}{2}\right) + \cosh^{-1}\left(\frac{\Tr g_2^\mathsf{R}}{2}\right) = \ell_2^{\mathsf{L}} + \ell_2^{\mathsf{R}}
\ee
respectively. A consequence of this fact is that neither $\xi_1$ nor $\xi_2$ will correspond to the generators of translations on the asymptotic boundary, contrary to the (2,0) and (3,0)-wormholes studied before. 

The conjugacy class corresponding to the horizon can be found by noticing that the closed geodesic must have trivial homology, so the element must be trivial when abelianized. The simplest horizon element obeying this constraint is
\be
\upgamma_1 =g_1 g_2 g_1^{-1}g_2^{-1},
\ee
and it can be checked that in the fundamental domain the geodesic connecting the fixed points of this element will contain the horizon segment homologous to the first asymptotic region. The other three segments are found by conjugation and they are
\be
\upgamma_2 = g_2 g_1^{-1}g_2^{-1}g_1 , \,\,\, \upgamma_3 = g_1^{-1}g_2^{-1}g_1 g_2  \text{ and  }\upgamma_4 = g_2^{-1}g_1 g_2 g_1^{-1}.
\ee
There are of course infinitely many more such horizon words $\gamma\,  \upgamma_1 \gamma^{-1}$, which fill the fractal geometry in the covering space.

The horizon's length is then computed as usual:
\begin{align}
\mathsf{h} &= \cosh^{-1}\left(\frac{\Tr \upgamma_\mathsf{L}}{2}\right) + \cosh^{-1}\left(\frac{\Tr \upgamma_\mathsf{R}}{2}\right)\\
&= \sum_{\mathsf{I}= \{\mathsf{L},\mathsf{R}\}} \cosh^{-1}\left[\cosh^2{\ell_2^{\mathsf{I}}} + (\cos{2\alpha_{\mathsf{I}}}\sinh^2{\ell_1^{\mathsf{I}}} - \cosh^2{\ell_1^{\mathsf{I}}})\sinh^2{\ell_2^{\mathsf{I}}}\right].
\end{align}
Just like for the 3-boundary wormhole, imposing $\Tr\gamma_{\mathsf{L},\mathsf{R}}>2$ imposes a constraint on the angles $\alpha_{\mathsf{L},\mathsf{R}}$, which equates to the requirement $\mathsf{h}>0$. Hence, the six moduli can be chosen to be the six positive numbers $(\ell_1^\mathsf{L}, \ell_1^\mathsf{R}, \ell_2^\mathsf{L}, \ell_2^\mathsf{R}, \ell_3^\mathsf{L}, \ell_3^\mathsf{R})$.

\subsubsection{Entanglement entropy}

We will now discuss the computation of entanglement entropies in this quotient. As previously mentioned, we fix 
\be
\alpha_{\mathsf{L}} = \alpha_{\mathsf{R}} = \frac{\pi}{2}, \quad \ell_1^{\mathsf{L}} = \ell_2^{\mathsf{L}},\quad \ell_1^{\mathsf{R}} = \ell_2^{\mathsf{R}}
\ee
so the remaining two moduli $(\ell_2^{\mathsf{L}},\ell_2^{\mathsf{R}})$ define a rotating (1,1)-wormhole where the two non-contractible circles of the torus behind the horizon are equal, with lengths $\ell_2^{\mathsf{L}} + \ell_2^{\mathsf{R}}$. The geometry is well defined by demanding the elements to be hyperbolic, which is equivalent to the constraint $\ell_2^{\mathsf{L},\mathsf{R}} > \sinh^{-1}(1)$.

Let $\mathsf{A}$ be a boundary region with endpoints $\mathbf{a}_1$ and $\mathbf{a}_2$ with coordinates $(t_1,\phi_1)$ and $(t_2,\phi_2)$ on the quotient's boundary. To compute the length of a geodesic as we have been doing so far, we need to find the image of these two points in the covering space, which can be done by translating some representative point $\mathbf{p}$ using the appropriate Killing vectors
\be
\mathbf{a}_1 = e^{(\phi_1 + t_1)\, \xi_1^{\mathsf{L}} }\mathbf{p}\,e^{(\phi_1 - t_1)\, \xi_1^{\mathsf{R}}}, \quad \mathbf{a}_2 = e^{(\phi_2 + t_2)\, \xi_1^{\mathsf{L}} }\,\mathbf{p}\,e^{(\phi_2 - t_2)\, \xi_1^{\mathsf{R}}}.
\ee
Note, however, that for the (1,1)-wormhole the choice for the representative is a bit more non-trivial since the boundary is composed by four disconnected components in the fundamental domain. This is not a problem, it suffices to chose one boundary component and translate the point $\mathbf{p}$ accordingly to find $\mathbf{a}_1$ and $\mathbf{a}_2$, however one must be careful to use the correct generator of translations associated to the particular boundary component chosen, which can be found by solving the corresponding matrix equations, i.e.
\be\label{eq:matrix-eq}
\xi_i = \frac{1}{2\pi} \log(-\upgamma_i).
\ee
For the symmetric (1,1)-wormhole in question, the representative points are simple and they correspond to the points located at $\phi = \pi/4, 3\pi/4, 5\pi/4$ and $7\pi/4$ on the covering space boundary. For a generic (1,1)-wormhole this is no longer true and the representative points must be found case by case depending on the choice of the six moduli. This can be done by computing the eigenvectors of the appropriate horizon element, which gives the two fixed points. The point $\mathbf{p}$ can be defined as the midpoint of the two fixed points, as explained in \cite{Maxfield:2014kra}.

\begin{figure}
    \centering
    \includegraphics[width=\textwidth]{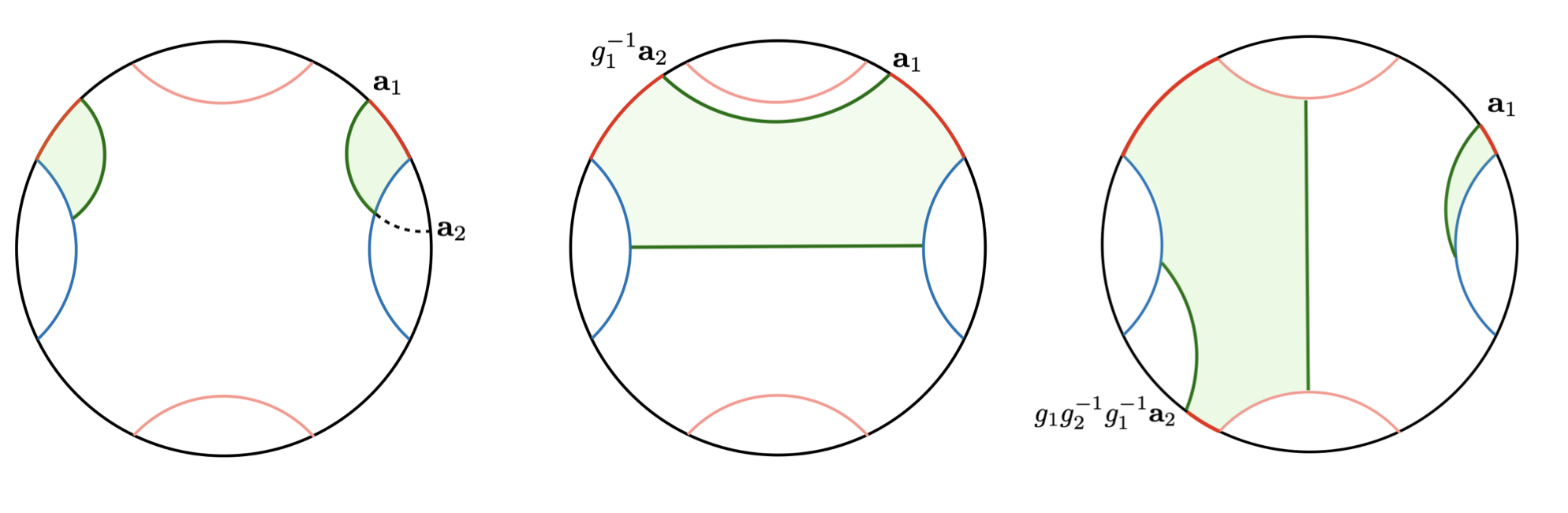}
    \caption{Three possible phases for the computation of $S(\mathsf{A})$ when $\mathbf{a}_1$ is in the first asymptotic component and $\Delta \phi < \pi$. For all drawings, we show the boundary region in red, the HRT surface in green with its homology region shaded. \textbf{[Left]} This is the trivial phase, identical to BTZ. As shown in the figure, it dominates when the boundary region $\mathsf{A}$ is small enough (typically whenever $\Delta\phi < \pi/2$). \textbf{[Center]} Phase where the dominant geodesic is in the homotopy class of $g_1^{-1}$, and as such we augment it with the closed geodesic in the conjugacy class of $g_1$. This phase dominates for large $\Delta \phi$ and $\mathbf{a}_2$ in the second asymptotic component. \textbf{[Right]} Phase where the dominant geodesic is in the homotopy class of $g_1 g_2^{-1} g_1^{-1}$. The geodesic must be augmented by the closed geodesic in the conjugacy class of $g_2$ (as can be seen by abelianizing the previous word). This phase dominates whenever $\Delta \phi$ is large enough and $\mathbf{a}_2$ is in the third asymptotic component. The horizon must be large enough for the two non-trivial phases to ever dominate.}
    \label{fig:phasesTorus}
\end{figure}

Similarly to the (3,0)-wormhole, the dominant group elements that compute $S(\mathsf{A})$ will depend not only on the size of $\mathsf{A}$ but on its position too (as well as on which connected component the computation was choosen to be performed). We find that even in the presence of rotation, the dominant phases agree with the one found in \cite{Maxfield:2014kra}.  We show in figure \ref{fig:phasesTorus} three possibilities for a region $\mathsf{A}$ with the endpoint $\mathbf{a}_1$ inside the first asymptotic component.

\subsubsection{Numerics}

Unlike the (3,0) wormhole case, it is not sufficient to test the inequalities for the simplest configuration of ordered regions with no gap. This is because the CFT state corresponding to the (1,1) geometry is pure, and any dihedral inequality identically vanishes by the property $S(A) = S(A^c)$. Considering MMI as an example, using purification symmetry we have $S(AB) + S(BC) + S(AC) = S(C) + S(A) + S(B)$.
Combinining it with $S(ABC) = 0$, we see that \eqref{eq:mmi} is saturated.

Therefore, we ran the numerics for Type-b regions, i.e.\ configurations of ordered regions with a gap between the first and last. This means that all terms appearing in any dihedral inequality correspond to the entanglement entropy of either one or two intervals.

\begin{figure}
    \centering
    \includegraphics[width=\textwidth]{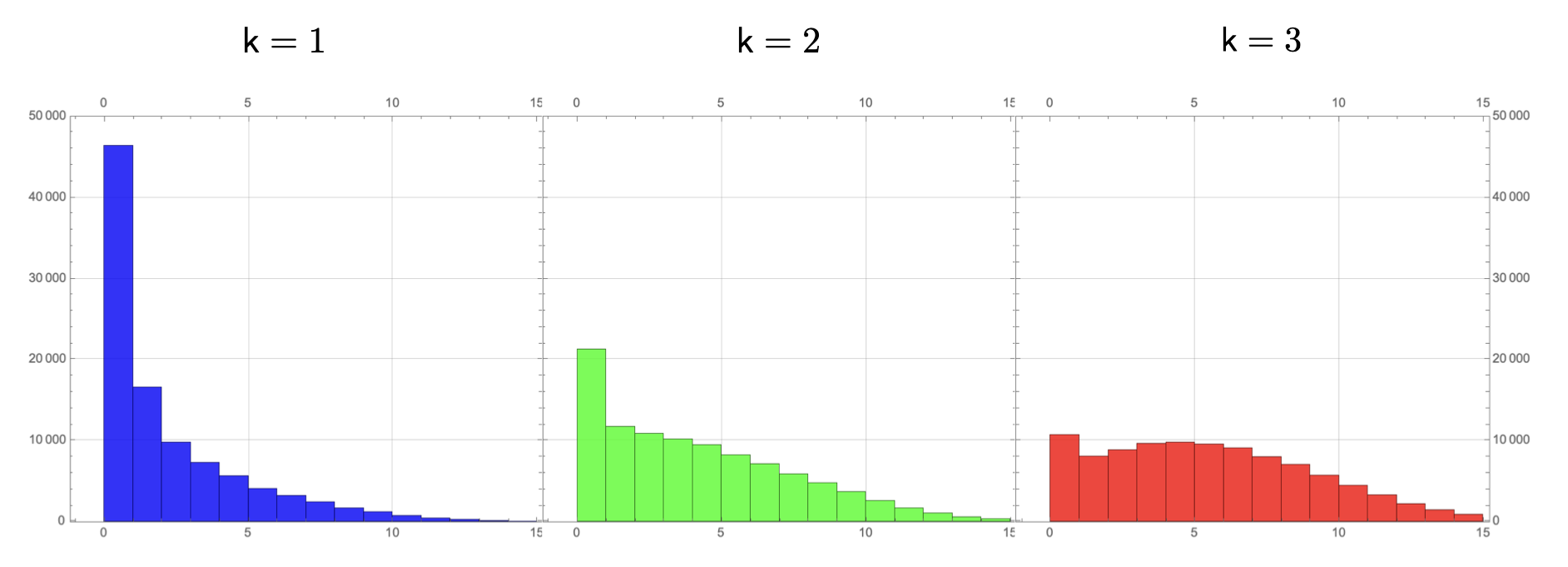}
    \caption{Distribution of $\Delta = $LHS-RHS for the first three dihedral inequalities when the bulk is a rotating symmetric (1,1) wormhole. For every iteration in the code the two moduli $(\ell_2^{\mathsf{L}},\ell_2^\mathsf{R})$ have been generated randomly in the range $(3.5,5)$.}
    \label{fig:histoTorus}
\end{figure}

To faithfully test the inequalities for this geometry, we want to make sure we sample the part of moduli space for which non-trivial phases exist, for otherwise one would simply be testing the inequalities for the BTZ case. This happens for large enough values of $\ell_2^{\mathsf{L},\mathsf{R}}$ (see \cite{Maxfield:2014kra} for a detailed analysis), so we randomly sample the moduli with a uniform distribution in the interval $(3.5,5)$ where we have the biggest chance of hitting non-trivial phases. It is worth mentioning that the code for the (1,1)-wormhole runs much slower than the previous two examples since the program needs to solve the matrix equation \ref{eq:matrix-eq} for each single trial. As with the previous two examples, the program stores the following array of information about each trial
    \be
    \{ \{ \ell_2^\mathsf{L}, \ell_2^\mathsf{R}\}, \textsc{points}, \Delta\}.
    \ee
and repeats the process for $n$ independent trials. We performed the following tests:
\begin{itemize}
    \item For the first three dihedral inequalities we collected $n = 10^5$ trials, and we show the results in figure \ref{fig:histoTorus}.
    \item For each of the remaining dihedrals up to $\mathsf{k} = 10$ (corresponding to $\mathsf{N} = 21$ parties + 1 gap), we collected $n = 10^4$ trials.
\end{itemize} 
No counterexample was found.

\section{General argument}
\label{sec:argument}

In this section we present an argument,\footnote{We use the word ``argument'', rather than ``proof'', as we will not attempt to be rigorous (although we're not aware of any reason that our argument couldn't be made rigorous).} for bulk spacetimes of type $(2,0)$, i.e.\ with spatial topology of a cylinder, that HRT entropies obey any RT inequality. Here we are not restricting ourselves to vacuum solutions, so the argument subsumes but goes beyond the numerical results for BTZ spacetimes described in subsection \ref{sec:BTZ} above. The argument builds on Czech-Dong's theorem that any RT inequality is obeyed by HRT if the bulk is simply connected \cite{Czech:2019lps}. As they suggested, to generalize that theorem to the case where the bulk is multiply connected, we pass to its universal covering space. However, whereas they considered the full lift of each region to the universal cover, which involves an infinite number of copies and therefore an IR-divergent entropy, we consider only a finite number of copies of each region. We use a finite covering space as a stepping stone to relate this entropy to the one on the original spacetime. In this argument, we will need to use the fact that a boundary region invariant under a bulk isometry admits an invariant HRT surface. We're not aware of a proof of this statement in the literature, so we devoted appendix \ref{sec:multipleHRT} to proving this explicitly (in any number of dimensions), and more generally to establishing some interesting facts about situations where a region admits multiple HRT surfaces. In subsection \ref{sec:other wormholes}, we will explain why the argument does not carry over immediately to spacetimes with other topologies, such as the $(3,0)$ and $(1,1)$ wormholes explored numerically in the previous section.

\subsection{(2,0) wormholes}
\label{sec:general(2,0)}

\begin{figure}
    \centering
    \includegraphics[width=.8\textwidth]{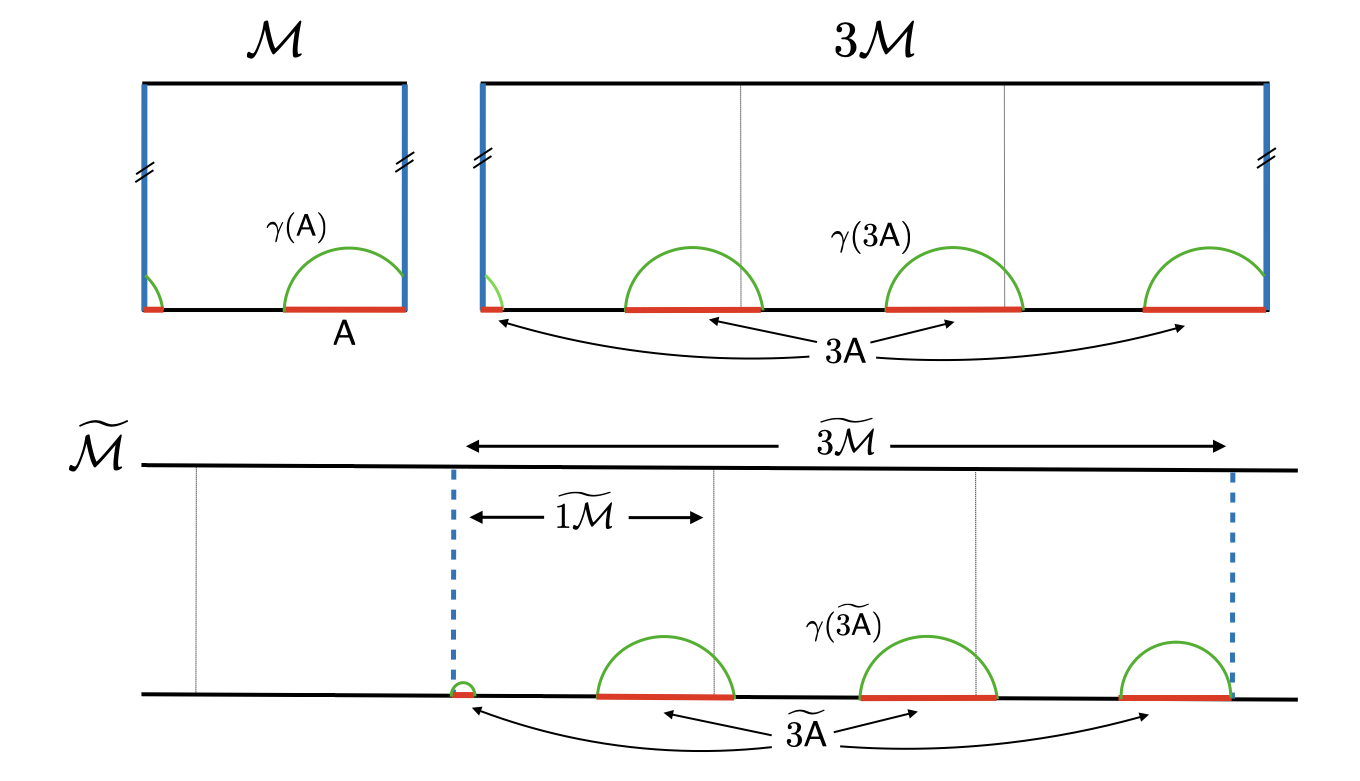}
    \caption{Schematic drawings of the relevant geometric objects entering the argument. For clarity, we draw only spatial slices, although of course these should be understood as embedded in their respective spacetimes. \textbf{[Top left]} Original spacetime $\mathcal{M}$. The top and bottom sides of the rectangle are the two asymptotic boundaries, and the two blue vertical sides are identified. The region $\mathsf{A}$ is shown in red and its HRT surface $\gamma(\mathsf{A})$ in green. 
    \textbf{[Top right]} $n$-fold cover $n\mathcal{M}$ (for the case $n = 3$). The two vertical blue sides are once again identified, making the $\mathbb{Z}_3$ symmetry of $3\mathcal{M}$ manifest, so it's clear that $|\gamma(3\mathsf{A})| = 3|\gamma(\mathsf{A})|$. \textbf{[Bottom]} Universal covering space $\widetilde{\mathcal{M}}$ (one should imagine the drawing continuing indefinitely to the left and right). Inside $\widetilde{\mathcal{M}}$ we show the fundamental domain $\widetilde{1\mathcal{M}}$ as well as $\widetilde{3\mathcal{M}}$. Crucially, the two sides are now cut along a hypersurface and are not identified. The HRT surface $\gamma(\widetilde{3\mathsf{A}})$ therefore will differ from $\gamma(3\mathsf{A})$ near the cut, since $\gamma(\widetilde{3\mathsf{A}})$ will get modified by the presence of the cut. In particular, the cut introduces two new entangling surfaces, and as explained in the text below $S(\widetilde{n\mathsf{A}})$ contains an extra UV-divergence of $(c/3)\ln(1/\epsilon)$. However, far from the cut the two surfaces $\gamma(\widetilde{3\mathsf{A}})$ and $\gamma(3\mathsf{A})$ coincide.
    }
    \label{fig:objects}
\end{figure}

\begin{figure}
    \centering
    \includegraphics[width=.8\textwidth]{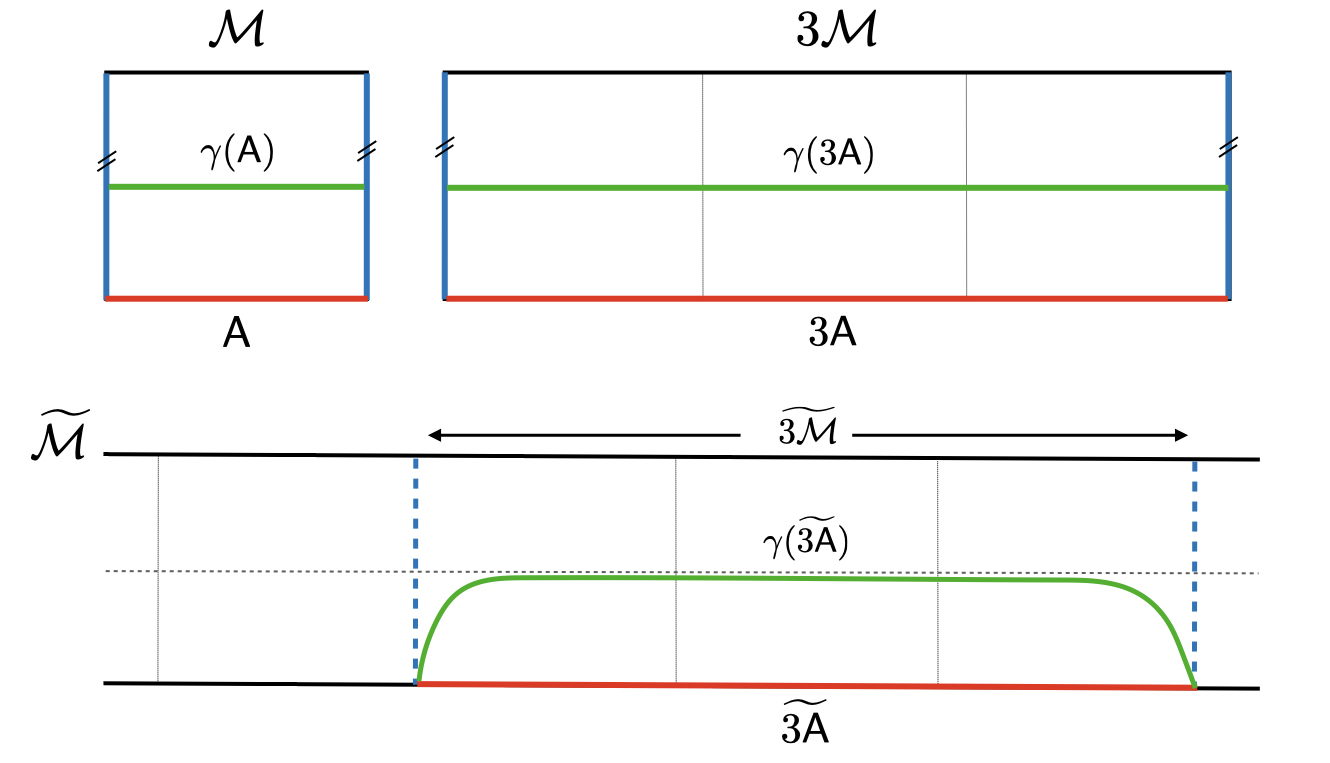}
    \caption{
    Same as figure \ref{fig:objects}, but where $\mathsf{A}$ (shown in red) is one entire boundary, so that its HRT surface $\gamma(\mathsf{A})$ (green) in $\mathcal{M}$ is a compact non-boundary anchored surface. While its lift $\gamma(n\mathsf{A})$ in $n\mathcal{M}$ is still compact, the HRT surface $\gamma(\widetilde{n\mathsf{A}})$ in $\widetilde{n\mathcal{M}}$ is boundary-anchored. For large $n$, far from the cut, it will hug (but not exactly coincide with) the lift of $\gamma(\mathsf{A})$ in the covering space $\widetilde{\mathcal{M}}$ (dotted horizontal line in $\widetilde{\mathcal{M}}$). The two surfaces, however, coincide in the limit $n \to \infty$, for a fixed position.
    }
    \label{fig:objects2}
\end{figure}

Let $\mathcal{M}$ be a $2+1$ dimensional spacetime obeying standard classical holographic assumptions (asymptotically locally AdS, globally hyperbolic, obeying the Einstein equation and null energy condition), with spatial topology of a cylinder. Fix a set of $\mathsf{N}$ elementary boundary regions $A,B,\ldots$ lying on a common Cauchy slice. Let $Q$ be an $\mathsf{N}$-party information quantity (a function of a joint state on $\mathsf{N}$ systems that is a linear combination of the entropies of the composite systems) such that $Q\ge0$ is an RT inequality.

The fundamental group of $\mathcal{M}$ is $\mathbb{Z}$. Before passing to the universal covering space, we consider the connected $n$-fold covering space, which we denote $n\mathcal{M}$, where $n$ is any positive integer. (See figures \ref{fig:objects}, \ref{fig:objects2} for an illustration of $n\mathcal{M}$ and the other relevant geometric objects entering into this discussion.) Note that any covering space of $\mathcal{M}$ obeys standard holographic assumptions, so the HRT formula can be sensibly applied. Given a composite boundary region $\mathsf{A}$ (including components on one or both boundaries), let $n\mathsf{A}$ be its lift to $n\mathcal{M}$; $n\mathsf{A}$ is a spatial boundary region for $n\mathcal{M}$. The HRT surface (or any HRT surface, if there is more than one) $\gamma(\mathsf{A})$ in $\mathcal{M}$ lifts to a candidate HRT surface for $n\mathsf{A}$ in $n\mathcal{M}$ (i.e.\ an extremal surface homologous to $n\mathsf{A}$). Conversely, there exist an HRT surface $\gamma(n\mathsf{A})$ that respects the $\mathbb{Z}_n$ isometry group of $n\mathcal{M}$ (see corollary \ref{invariant} in appendix \ref{sec:multipleHRT}), and therefore descends to a candidate HRT surface for $\mathsf{A}$. So we have
\be\label{S(nA)1}
S(n\mathsf{A})=nS(\mathsf{A})\,.
\ee

Now let $\widetilde{\mathcal{M}}$ be the universal covering space of $\mathcal{M}$. Fix a fundamental domain $\widetilde{1\mathcal{M}}$ for $\mathcal{M}$ in $\widetilde{\mathcal{M}}$. It will be convenient to choose $\widetilde{1\mathcal{M}}$ so that its boundary does not intersect any of the entangling surfaces (the boundaries of the elementary regions, or more precisely their lifts to $\widetilde{\mathcal{M}}$). Let $\widetilde{n\mathcal{M}}$ be the union of $n$ consecutive copies of $\widetilde{1\mathcal{M}}$. $\widetilde{\mathcal{M}}$ is also the universal cover of $n\mathcal{M}$, and $\widetilde{n\mathcal{M}}$ is a fundamental domain for the latter. In other words, $\widetilde{n\mathcal{M}}$ is equal to $n\mathcal{M}$ cut open along some hypersurface.

Any composite region $\mathsf{A}$ has an image in $\widetilde{n\mathcal{M}}$, consisting of $n$ copies, which we denote $\widetilde{n\mathsf{A}}$. (And of course we have $\widetilde{n(AB)}=(\widetilde{nA})\cup(\widetilde{nB})$, etc.) Under the covering projection from $\widetilde{\mathcal{M}}$ to $n\mathcal{M}$, the HRT surface $\gamma(\widetilde{n\mathsf{A}})$ in $\widetilde{\mathcal{M}}$ maps to an extremal surface in $n\mathcal{M}$. That surface is not necessarily equal to $\gamma(n\mathsf{A})$, due to the effect of the cut, and therefore in general $S(\widetilde{n\mathsf{A}})\neq S(n\mathsf{A})$. However, for $n\gg1$, far away from the cut, we expect its effect to be negligible. (This statement seems intuitively clear to us, but we have not attempted to prove it formally.) Hence the difference in their areas should be subleading in $n$:
\be\label{areadiff}
\lim_{n\to\infty}\frac{S(\widetilde{n\mathsf{A}}) - S(n\mathsf{A})}n=0\,.
\ee
Given \eqref{S(nA)1}, this implies
\be\label{S(nA)}
\lim_{n\to\infty}\frac{S(\widetilde{n\mathsf{A}})}n = S(\mathsf{A})\,.
\ee
Hence,
\be\label{nQ}
\lim_{n\to\infty}\frac{Q(\widetilde{nA},\widetilde{nB},\ldots)}n = Q(A,B,\ldots)\,.
\ee
Since $\widetilde{\mathcal{M}}$ is a simply-connected spacetime, according to Czech-Dong \cite{Czech:2019lps}, $Q(\widetilde{nA},\widetilde{nB},\ldots)\ge0$, so $Q(A,B,\ldots)\ge0$.

Actually, there is a caveat to \eqref{areadiff} and \eqref{S(nA)}, which is that the two entropies appearing in each of these equations may differ by a UV-divergent amount. Specifically, due to the cut, $S(\widetilde{n\mathsf{A}})$ contains an extra contribution of $n_{\mathsf{A}}(c/3)\ln(1/\epsilon)$, where $\epsilon$ is the UV cutoff and $n_{\mathsf{A}}$ is the number of components of $\mathsf{A}$ that are cut by the fundamental domain boundary. Hence there is an order-of-limits issue: while \eqref{areadiff} and \eqref{S(nA)} are correct at fixed $\epsilon$, we need to take $n\gg\ln(1/\epsilon)$ to see the convergence. However, this issue does not apply to \eqref{nQ}: All RT inequalities are balanced, meaning that, for any elementary region, the total coefficient of all the terms in $Q$ involving that region vanishes.\footnote{More precisely, all primitive RT inequalities --- those that cannot be written as a positive linear combination of other RT inequalities --- are balanced, except Araki-Lieb, $S(AB)+S(A)\ge S(B)$.} This implies that the extra divergences due to the cut cancel in $Q(\widetilde{nA},\widetilde{nB},\ldots)$, so the difference  $Q(\widetilde{nA},\widetilde{nB},\ldots)/n-Q(A,B,\ldots)$ is UV-finite for all $n$.

\subsection{Other topologies?}
\label{sec:other wormholes}

The general construction of the proof generalizes straightforwardly to an arbitrary wormhole. However, as we will explain, because of a key difference in the structure of the fundamental group in cases other than $(2,0)$, the estimate \eqref{areadiff} does not obviously carry over.

As above, we start with finite covers. Given an $n$-fold cover $n\mathcal{M}$ of $\mathcal{M}$, its deck transformations are isometries. If the cover is normal, then $\mathcal{M}$ is the quotient of $n\mathcal{M}$ by the group of deck transformations, which plays the role of $\mathbb{Z}_n$ in the paragraph above \eqref{S(nA)1}, and by the same argument, we again have \eqref{S(nA)1}. 

\begin{figure}
    \centering
    \includegraphics[width=.3\textwidth]{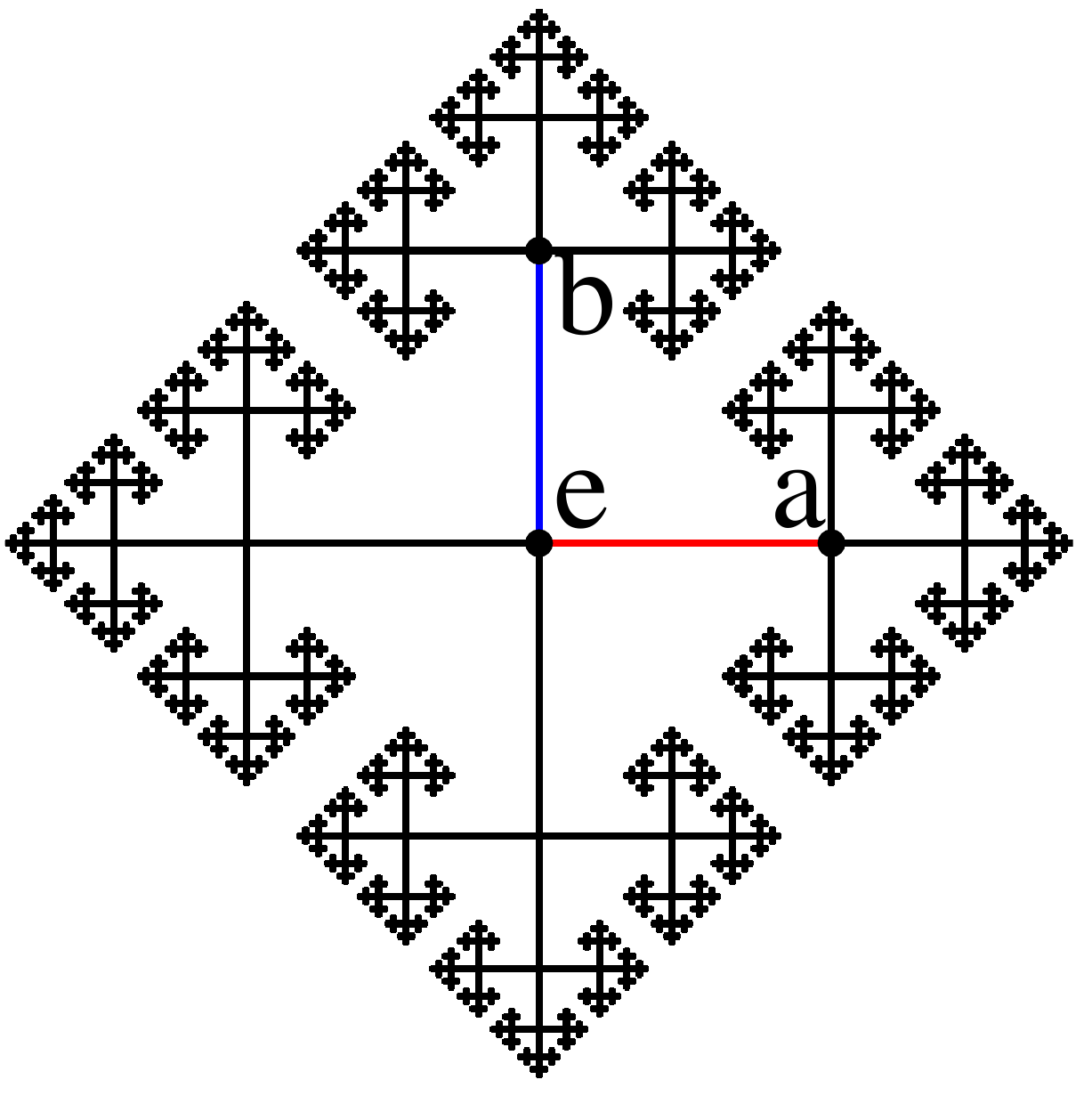}
    \caption{The Cayley graph $G$ for $F_2$, the free group with 2 generators $a,b$. $F_2$ is the fundamental group for both $(3,0)$ and $(1,1)$ wormholes. (Figure from Wikipedia.)}
    \label{fig:Cayley}
\end{figure}

We now turn to the universal cover. For an $(n',g)$ wormhole, the fundamental group $\pi_1(\mathcal{M})$ is that of a Riemann surface with genus $g$ and $n'\ge1$ punctures, namely the free group $F_r$ of rank $r=2g+n'-1$. The universal cover $\widetilde{\mathcal{M}}$ can be constructed in terms of the Cayley graph $G$ of its fundamental group, an infinite tree of coordination number $2r$; see figure \ref{fig:Cayley} for the Cayley graph of $F_2$, the fundamental group of both $(3,0)$ and $(1,1)$ wormholes. Each vertex of $G$ corresponds to a fundamental domain, and the fundamental domains are glued together along the edges of the graph.

Again, $n\mathcal{M}$ can be mapped bijectively to a subset $\widetilde{n\mathcal{M}}$ of $\widetilde{\mathcal{M}}$ consisting of $n$ contiguous fundamental domains, corresponding to a set of $n$ vertices of $G$; again, $\widetilde{n\mathcal{M}}$ is equal to $n\mathcal{M}$ cut open along some hypersurface. Again, we wish to compare the entropy of the lift of a given composite region $\mathsf{A}$ in $\widetilde{n\mathcal{M}}$ to that in $n\mathcal{M}$, in other words to compute $S(\widetilde{n\mathsf{A}})-S(n\mathsf{A})$. And again, for $n\gg1$, far away from the cut, we expect the effect of the cut to be negligible. The next step, however, namely the statement that ``the difference in their areas should be subleading in $n$'', does not follow anymore. The reason is that the fraction of $\widetilde{n\mathcal{M}}$, and therefore the fraction of the HRT surface $\gamma(\widetilde{n\mathsf{A}})$, that is \emph{close} to the cut (in any sense of \emph{close}), and therefore where the effect of the cut may \emph{not} be negligible, is not subleading. This is because of the hyperbolicity of $G$: for any family $s_n$ of sets of $n$ vertices, the fraction of vertices in $s_n$ within a given distance of its complement does not go to 0 as $n\to\infty$ (see figure \ref{fig:Cayley}). For the $(2,0)$ case, on the other hand, the Cayley graph is linear, so for connected subsets, that fraction does go to 0.

The numerical evidence collected in section \ref{sec:numerical} nonetheless strongly supports the validity of RT inequalities for $(3,0)$ and $(1,1)$ wormholes, and by extrapolation for all wormholes. If this is true, then either there exists a clever variation on the above argument that evades the issue with the hyperbolicity of the Cayley graph, or the inequalities are simply valid for some completely different reason.

\acknowledgments

We would like to thank B. Czech and X. Dong for helpful correspondence and conversations, as well as A. Wall for encouragement to publish this work and B. Czech for valuable comments on a draft. B.G.W., G.G., and M.H. were supported in part by the Department of Energy through awards DE-SC0009986 and QuantISED DE-SC0020360, and in part by the Simons Foundation through the \emph{It from Qubit} Simons Collaboration. B.G.W was also supported by the AFOSR under FA9550-19-1-0360. V.H. was supported in part by the Department of Energy through awards DE-SC0009999 and QuantISED DE-SC0020360, and by funds from the University of California.
This research was supported in part by grant NSF PHY-2309135 to the Kavli Institute for Theoretical Physics (KITP), where part of this work was completed. This work was also performed in part at the Aspen Center for Physics, which is supported by National Science Foundation grant PHY-2210452. We are also grateful to the Perimeter Institute, the Centro de Ciencias de Benasque Pedro Pascual, and the Yukawa Institute for Theoretical Physics, where part of this work was completed.

\appendix

\section{Calculation of entanglement entropies in AdS$_3$ quotients}\label{sec:quotients}

In this appendix we give a brief review of \cite{Maxfield:2014kra}, where a framework to compute holographic entanglement entropies in pure 3d gravity was introduced. In particular we will be interested in computing lengths of boundary-anchored geodesics in the quotient spacetime and systematically imposing the homology condition.

\subsection*{AdS$_3$ as a group manifold}

Anti de Sitter space is the maximally symmetric Einstein manifold with negative constant curvature. In 2+1 dimensions, it can be described as the submanifold
\begin{equation}
    -U^2 -V^2 + X^2 + Y^2 = -1
\end{equation}
embedded in $\mathbb{R}^{2,2}$. In the above we have set the AdS scale to one. Equivalently, we can view this locus of points as the group manifold SL(2,$\mathbb{R}$) by combining the coordinates $(U,V,X,Y)$ into a $2\times 2$ matrix with unit determinant, i.e. 
\begin{equation}
    \mathbf{M} = 
    \begin{pmatrix}
    U+X &\, Y-V\\
    Y+V &\, U-X
    \end{pmatrix}, \quad \det \mathbf{M} = 1.
\end{equation}
It is then clear that the isometries of this manifold are elements of the group
\begin{equation}
    (\text{SL}(2,\mathbb{R}) \times \text{SL}(2,\mathbb{R}))/ \mathbb{Z}_2 
\end{equation}
acting on $\mathbf{M}$ by left and right multiplication, that is $(g_L, g_R) \in \text{SL}(2,\mathbb{R}) \times \text{SL}(2,\mathbb{R})$ with action defined as $\mathbf{M} \mapsto g_L \mathbf{M} g_R^{\top}$. The quotient by $\mathbb{Z}_2$ is there because an overall sign results in the same transformation. 

We will be particularly interested in describing the asymptotic boundary in this group manifold picture. To do that we introduce coordinates $(r,t,\phi)$ which are related to the embedding coordinates by
\begin{align}
U &= \sqrt{1+r^2}\cos{t},\\
V &=\sqrt{1+r^2}\sin{t},\\
X &= r \cos{\phi},\\
Y &= r \sin{\phi}.
\end{align}
with $0<r<\infty$, and both $\phi$ and $t$ have period $2\pi$. To obtain the more standard AdS$_3$ spacetime one considers the covering space of this hyperboloid by unwrapping the time direction. In these coordinates the metric takes the standard form
\be
\dd s^2 = - (1+ r^2) \dd t^2 + \frac{\dd r^2}{1+r^2} + r^2 \dd \phi^2.
\ee
In these coordinates, points in $\mathbf{M}$ then take the form
\be
\begin{pmatrix}
    \sqrt{1+r^2}\cos{t} + r \cos{\phi} & \; \; \; r\sin{\phi}-\sqrt{1+r^2}\sin{t}\\
    \sqrt{1+r^2}\sin{t} + r\sin{\phi} & \; \; \; \sqrt{1+r^2}\cos{t}- r\cos{\phi} 
\end{pmatrix},
\ee
and the asymptotic boundary $\partial \mathbf{M}$ is found by taking $r$ large
\be
\partial \mathbf{M}\sim
\begin{pmatrix}
 \cos{t}+ \cos{\phi}& \sin{\phi}- \sin{t}\\
 \sin{t} + \sin\phi &\cos{t} - \cos\phi 
\end{pmatrix} = 2\begin{pmatrix}
\cos{\frac{t + \phi}{2}}\cos{\frac{t - \phi}{2}} &-\cos{\frac{t + \phi}{2}}\sin{\frac{t - \phi}{2}} \\
 \sin{\frac{t + \phi}{2}}\cos{\frac{t - \phi}{2}}&-\sin{\frac{t + \phi}{2}}\sin{\frac{t - \phi}{2}}.
\end{pmatrix} 
\ee
We removed a common radial factor which blows up. Equivalently, the above set of matrices can be written as a product of unit vectors $\vec{x}, \vec{y}$ as follows
\be
2\begin{pmatrix}
\cos{\frac{t + \phi}{2}}\cos{\frac{t - \phi}{2}} &-\cos{\frac{t + \phi}{2}}\sin{\frac{t - \phi}{2}} \\
 \sin{\frac{t + \phi}{2}}\cos{\frac{t - \phi}{2}}&-\sin{\frac{t + \phi}{2}}\sin{\frac{t - \phi}{2}} 
\end{pmatrix} 
= 2
\begin{pmatrix}
\cos{\frac{t + \phi}{2}}\\
\sin{\frac{t + \phi}{2}}
\end{pmatrix}
\begin{pmatrix}
\cos{\frac{t - \phi}{2}}&
-\sin{\frac{t - \phi}{2}}
\end{pmatrix}
\propto \vec{x} \vec{y}^{\top}.
\ee
When studying entropy inequalities, we will often be interested in points lying on the $t = 0$ slice on the boundary. These are the locus of points $\vec{x}\vec{x}^\top$, i.e.\ when $\vec{x}=\vec{y}$. We will also be interested in points on the asymptotic boundary that are unmoved when acted by an isometry. These are called fixed points and are defined as
\be
g_L \mathbf{p} g_R^{T} \propto  \mathbf{p} \quad\Rightarrow\quad (g_L \vec{x}) (g_R \vec{y})^\top \propto \vec{x}\vec{y}^{\top}
\ee
or, in other words, they are points where $\vec{x}$ and $\vec{y}$ are the eigenvectors of $g_L$ and $g_R$ respectively, where we left out an overall proportionality since it doesn't affect the asymptotic boundary.

\subsection*{Geodesic lengths in AdS$_3$}

Since we are working with entanglement entropies, we will be interested in computing the geodesic length between points on the asymptotic boundary. Let us begin by computing the geodesic distance between two arbitrary points $\mathbf{p}$ and $\mathbf{q}$ which are spacelike separated. This is easiest to compute starting from the embedding coordinates since geodesics are straight lines when viewed in $\mathbb{R}^{2,2}$, with length $\ell$ given by
\begin{align}
\cosh{\ell(\mathbf{p},\mathbf{q})} &= U_1 U_2 + V_1 V_2 - X_1 X_2 - Y_1 Y_2 \\ 
&= \frac{\Tr(\mathbf{p}^{-1}\mathbf{q})}{2}.
\end{align}
If we send $\mathbf{p}$ and $\mathbf{q}$ on the asymptotic boundary we can pull a radial factor out each matrix, obtaining
\begin{align}
\cosh{\ell(\mathbf{p},\mathbf{q})} &= \frac{r^2 \Tr(R_{\perp}\mathbf{p}^{\top}R^{\top}_{\perp}\mathbf{q})}{2}
\end{align}
where $R_{\perp}$ is the rotation matrix by $\pi/2$ used to implement the inverse since once we pull out the radial factor the matrices become singular. The regularized length between boundary points can then be written as
\begin{align}\label{eq:lengthAdS3}
    \ell(\mathbf{p}, \mathbf{q}) &= \cosh^{-1}\left(\frac{r^2\, \Tr(R_{\perp}\mathbf{p}^{\top}R^{\top}_{\perp}\mathbf{q})}{2}\right)\\
    &= \log[\Tr(R_{\perp}\mathbf{p}^{\top}R^{\top}_{\perp}\mathbf{q})] + \dots\\
    &=\log{[(R_{\perp}\vec{x}_1\cdot \vec{x}_2)(R_{\perp}\vec{y}_1 \cdot \vec{y}_2)]}.
\end{align}
where we have disregarded the radial factor and kept only the finite parts.

\subsection*{Quotients}

In 2+1 dimensions, there are no propagating degrees of freedom, so any Einstein manifold with $\Lambda <0$ is locally AdS$_3$. However, there may still be interesting structure globally, in particular made by taking quotients of $\text{AdS}_3$ by an appropriate subgroup $\Gamma$ of its isometry group. More specifically, $\Gamma$ will be a discrete subgroup of $\text{SL}(2,\mathbb{R})\times \text{SL}(2,\mathbb{R})$ defined by a finite number of generators with no relations, i.e. a free group of some rank $r$.  

To obtain a quotient spacetime that is free of pathologies, there are a few things to take care of. As previously mentioned, elements of $\Gamma$ will have fixed points which will form singularities in the quotient space $\text{AdS}_3/\Gamma$ . Therefore, we first remove from $\text{AdS}_3$ all points to the future and past of the fixed points of all elements $\gamma \in \Gamma$. To obtain a spacetime free of closed timelike curves, we must quotient only points that are spacelike separated. In practice this is enacted by demanding the elements of $\Gamma$ to be hyperbolic, which are the elements that are "spacelike" separated from the identity, when one places elements of $\text{SL}(2,\mathbb{R})$ in a Penrose diagram. It is convenient to represent these hyperbolic elements as exponentials of elements of the Lie algebra. The $\mathfrak{sl}(2,\mathbb{R})$ Lie algebra can be seen as a three-dimensional Lorentz vector space with coordinates $(t,x,z)$ and Killing norm given by $-t^2 + x^2 + z^2$. An element of this Lie algebra then takes the form
\be\label{eq:lie-algebra}
\frac{1}{2}
\begin{pmatrix}
    z  & x-t\\
    x+t & -z
\end{pmatrix}.
\ee

\subsection*{Lengths in the quotient space}

In this paper we will be interested in computing the entanglement entropy of an interval on the boundary of the quotient space $\text{AdS}_3/\Gamma$. Let $\mathbf{p}$ and $\mathbf{q}$ be two points on the boundary of the quotient space. Since $\text{AdS}_3$ is simply connected, the quotient space has fundamental group isomorphic to $\Gamma$. In each fixed endpoint homotopy class there exists a unique geodesic. Thus, for fixed $\mathbf{p}$ and $\mathbf{q}$ there is one geodesic for each element $\gamma \in \Gamma$. To compute the length of each geodesic, we move into the covering space $\text{AdS}_3$, pick two appropriate representative points for $\mathbf{p}$ and $\mathbf{q}$ and use \eqref{eq:lengthAdS3}. In practice one can define $\hat{\mathbf{p}}$ and $\hat{\mathbf{q}}$ to be in the fundamental domain in the covering space, then move $\hat{\mathbf{q}}$ to any $\gamma_L \mathbf{q} \gamma_R^{\top}$ while keeping $\hat{\mathbf{p}}$ fixed --- doing this will compute the length of the geodesic in the $\gamma$ homotopy class in the quotient. Therefore, for $\mathbf{p}$ and $\mathbf{q}$ boundary points on the quotient the length of the geodesic in the $\gamma$ homotopy class is
given by
\begin{align}\label{eq:lengthQuotient}
\ell(\mathbf{p},\mathbf{q},\gamma) &= \log[\Tr(R_{\perp}\mathbf{p}^{\top}R^{\top}_{\perp} \gamma_L\mathbf{q} \gamma_R^{\top})]\\
&= \log{[(R_{\perp}\vec{x}_1\cdot \gamma_{L} \vec{x}_2)(R_{\perp}\vec{y}_1 \cdot \gamma_{R}
\vec{y}_2)]}
\end{align}
with $\gamma = (\gamma_L,\gamma_R) \in \Gamma$. Similarly, closed geodesics will also be in one-to-one correspondence with conjugacy classes of $\Gamma$. In the covering space a closed geodesic gets lifted to a geodesic connecting fixed points, and its length is
\be
\ell(\gamma) = \cosh^{-1}\left(\frac{\Tr \gamma_L}{2}\right) + \cosh^{-1}\left(\frac{\Tr \gamma_R}{2}\right).
\ee

\subsection*{The homology condition}\label{homology}

According to the RT and HRT prescriptions, one must perform the extremization among surfaces that are spacelike homologous to the boundary region. One way to say this is that there exists a codimension-1 spacelike region $\mathsf{R}$ whose boundary is the union of the HRT surface $\gamma_{\text{HRT}}$ and the boundary region $\mathsf{A}$, i.e. $\partial \mathsf{R} = \gamma_{\text{HRT}} \cup \mathsf{A}$. Another way of saying this is that the surface $\gamma_{\text{HRT}} \cup \mathsf{A}$ has trivial homology. By the Hurewicz theorem, the first homology group is the abelianization of the fundamental group. Since the fundamental group of the quotient spacetime is $\Gamma$ and since to each homotopy class $\gamma \in \Gamma$ there is an associated geodesic $\mathsf{g}_{\gamma}$, one can compute the homology of $\mathsf{g}_{\gamma}$ by simply looking at the abelianization of $\gamma$. If the resulting word is not trivial, the geodesic must be augmented by the closed geodesic associated to the element which trivializes the homology of the original geodesic once abelianized.

\section{Degenerate HRT surfaces \& symmetry breaking}
\label{sec:multipleHRT}

In this appendix we discuss situations in which a boundary region admits multiple HRT surfaces. This occurs when the bulk metric and boundary region are fine-tuned so that competing candidate surfaces have the same area. The motivation in the context of this paper has to do with spontaneous breaking of symmetries; in particular, in the discussion of subsection \ref{sec:general(2,0)}, we needed the existence of an HRT surface that does \emph{not} a break a symmetry shared by the bulk and boundary region; we will prove this fact here. The results of this appendix are likely known (or perhaps obvious) to experts on holographic entanglement entropy; however, as far as we know they have not yet appeared in written form.

Throughout this appendix we fix a bulk spacetime and a boundary region $\mathsf{A}$. We will also make a very mild genericity assumption, that for every HRT surface at least one of the principal eigenvalues of the null diagonal components of the Jacobi operator is non-zero; essentially this rules out continuos families of HRT surfaces (see the proof of lemma \ref{lemma:allHRT}). The discussion is general with respect to dimensionality; unlike in the rest of the paper, we do not restrict ourselves to $2+1$ dimensions.

Our main result will be the construction of a nested family of HRT surfaces. For this purpose, the maximin formulation of HRT \cite{Wall:2012uf} is useful. Given a bulk Cauchy slice $\sigma$ containing the entangling surface $\partial\mathsf{A}$, we define $\mathsf{A}_\sigma:= D(A)\cap \partial\sigma$; $\sigma$ is a maximin slice if the minimal area over surfaces in $\sigma$ homologous to $\mathsf{A}_\sigma$ equals $S(\mathsf{A})$. We then have the following lemma:
\begin{lemma}\label{lemma:allHRT}
Any maximin slice $\sigma$ contains all HRT surfaces. 
\end{lemma}
\begin{proof}
Otherwise, an HRT surface $\gamma$ not contained in $\sigma$ can be projected onto it, giving (by the usual focusing argument) a surface $\gamma'$ on $\sigma$ that is no larger than $\gamma$. Generically, the expansion is negative somewhere on the null congruence between $\gamma$ and $\gamma'$, so that $\gamma'$ has smaller area than $\gamma$, a contradiction.

However, it is not always the case that the expansion is negative somewhere on the null congruence; for example, for the bifurcation surface of a static black hole, the expansion vanishes everywhere on its null congruence, which is the event horizon. For this non-generic case, a more complicated argument is required. Assume for simplicity that $\sigma$ is nowhere to the past of $\gamma$; if it is nowhere to the future, we just switch future and past in the argument, while if it is partly in the future and partly in the past then a small adaption is required that we leave as an exercise to the reader. Fix future-directed null coordinates $u,v$ for the normal bundle to $\gamma$, such that the induced metric on the normal space is $ds^2=-2du\,dv$. Given a function $\phi$ on $\gamma$, let $\gamma_\phi$ denote the surface at the position $u=0$, $v=\epsilon\phi$, where $\epsilon$ is a fixed small parameter. By assumption, $\gamma_\phi$ has the same area as $\gamma$. Its expansion in the positive $u$ direction is
\be
\theta = -\epsilon {J^v}_v\phi\,,
\ee
where $J$ is the Jacobi operator of $\gamma$. We will not need the full expression for ${J^v}_v$, but it is an elliptic operator taking the general form
\be\label{Jform}
{J^v}_v = -\nabla^2+2t^i\partial_i+V\,,
\ee
where $\nabla^2$ is the scalar Laplacian, $t^i$ is a vector field, and $V$ is a function, all on $\gamma$. Due to the $\partial_i$ term, this operator is not self-adjoint. Nonetheless, according to lemma 4.1 of \cite{Andersson:2007fh}, any operator of the form \eqref{Jform} admits a real eigenvalue $\lambda$ (called the principal eigenvalue) whose corresponding eigenfunction $\hat\phi$ is everywhere non-negative.\footnote{Actually, lemma 4.1 of \cite{Andersson:2007fh} requires $\gamma$ to be connected and closed. However, based on the description of the proof, the closedness assumption can be replaced with $\gamma$ being compact with a Dirichlet boundary condition imposed on $\phi$. We can also drop the connectedness assumption, at the cost of the principal eigenvalue potentially being degenerate and the principal eigenfunction potentially being non-negative rather than positive, both of which are fine for our application.}

We can do a similar construction in the $u$ direction. We will make the genericity assumption that the principal eigenvalues of ${J^v}_v$ and ${J^u}_u$ are not both zero. Some assumption of this kind is necessary for the lemma to hold. For example, if the metric in a neighborhood of $\gamma$ is a product, $ds^2=-2du\,dv+ds_{\gamma}^2$, then both ${J^v}_v$ and ${J^u}_u$ would have vanishing principal eigenvalue, and indeed the lemma would be false.

Assume without loss of generality that ${J^v}_v$ has a nonzero principal eigenvalue $\lambda$. We cannot have $\lambda<0$, for then the expansion of $\gamma_{\hat\phi}$ in the negative $u$ direction would be non-positive (and somewhere negative), hence its projection onto any slice containing $\gamma$ would be smaller than $\gamma$, contradicting the fact that $\gamma$ is maximin, hence maximal on some slice. So $\lambda>0$. Therefore, the expansion of $\gamma_{\hat\phi}$ in the positive $u$ direction is everywhere non-positive, and somewhere negative. Its projection onto $\sigma$ therefore has area smaller than $S(\mathsf{A})$, at last giving the contradiction we were looking for.
\end{proof}

Each HRT surface has a corresponding candidate entanglement wedge, the causal domain of the associated homology region, which we will call its ``wedge'' for short.\footnote{In the presence of degenerate HRT surfaces, it is natural to expect that the smallest candidate wedge, i.e.\ the one that is contained in all the others, is the true entanglement wedge, in the sense that bulk observables localized in it are contained in the CFT subalgebra associated with $\mathsf{A}$. Indeed, Theorem \ref{lemma:nested} shows that there exists a wedge that is contained in all the others.} While the wedges in a given collection are not necessarily nested, we will now show that we can construct a nested set.

\begin{theorem}\label{lemma:nested}
Suppose $\mathsf{A}$ admits a set of $N$ HRT surfaces $\gamma_i$, with corresponding wedges $W_i$. Then there exists a set of $N$ HRT surfaces $\tilde\gamma_a$, such that $\bigcup_a\tilde\gamma_a=\bigcup_i\gamma_i$, that are nested, i.e.\ whose wedges $\tilde W_a$ satisfy
\be
\tilde W_N\subseteq \tilde W_{N-1}\subseteq\cdots\subseteq \tilde W_1\,,
\ee
and such that
\be
\tilde W_1=\bigcup_iW_i\,,\qquad
\tilde W_N=\bigcap_iW_i\,.
\ee
(Note that the $\tilde\gamma_a$ are not necessarily all distinct, even if the $\gamma_i$ are.)
\end{theorem}
\begin{proof}
Let $\sigma$ be a maximin slice. Let $r_i$ be the homology region on $\sigma$ for $\gamma_i$. For $a=1,\cdots,N$, let $\tilde r_a$ be the set of points on $\sigma$ contained in at least $a$ of the regions $r_i$. Thus, $\tilde r_1$ is their union while $\tilde r_N$ is their intersection. Each region intersects $\partial\sigma$ exactly on $\mathsf{A}_\sigma$. Hence the bulk part of its boundary, $\tilde\gamma_a:=\partial\tilde r_a\setminus\mathsf{A}$, is homologous to $\mathsf{A}$.
The total area of these surfaces is bounded by the total area of the original ones $\gamma_i$:
\be\label{suma}
\sum_a|\tilde\gamma_a|\le\sum_i|\gamma_i|=NS(\mathsf{A})\,.
\ee
(The reason it is not an equality is that two of the $\gamma_i$ could partially coincide, with the regions $r_i$ being locally complementary, so that in their union that part of the surfaces would not appear.) On the other hand, the minimal area on $\sigma$ is $S(\mathsf{A})$, so
\be\label{eacha}
|\tilde\gamma_a|\ge S(\mathsf{A})\,.
\ee
Together, \eqref{suma}, \eqref{eacha} imply that each $\tilde\gamma_a$ has area $S(\mathsf{A})$, and is therefore itself an HRT surface. (Hence the inequality in \eqref{suma} is saturated, which actually rules out the possibility of locally complementary homology regions.) These are clearly nested, $\tilde r_N\subseteq \tilde r_{N-1}\subseteq\cdots\subseteq\tilde r_1$, so their causal domains, $\tilde W_a$, are nested as well.
\end{proof}

The statement we needed in subsection \ref{sec:general(2,0)} about the existence of a symmetric HRT surface is a corollary of this lemma:
\begin{corollary}\label{invariant}
Suppose $\mathsf{A}$ is invariant under a finite group $G$ of isometries of the bulk. Then it admits an HRT surface that is invariant under $G$.\footnote{One can also consider a slightly more general notion of symmetry, which includes isometries that exchange $\mathsf{A}$ with $\mathsf{A}^c$. Interestingly, it \emph{can} happen that there is no invariant HRT surface for such a symmetry. For an example, consider a spherically symmetric star at the center of an asymptotically AdS spacetime (of any dimension), and let $\mathsf{A}$ cover one hemisphere of the boundary. In this configuration, there will be only two HRT surfaces, which pass on either side of the star and are exchanged by the $\mathbb{Z}_2$ reflection that exchanges $\mathsf{A}$ with $\mathsf{A}^c$.}
\end{corollary}
\begin{proof}
Let $\gamma_1$ be an HRT surface for $\mathsf{A}$, and let the surfaces $\{\gamma_i\}$ be obtained by acting on $\gamma_1$ with the elements of $G$. Then the surfaces $\tilde\gamma_a$ constructed in lemma \ref{lemma:nested} are invariant under $G$.
\end{proof}

Lemma \ref{lemma:nested} has another corollary that strongly restricts the possible forms of degenerate HRT surfaces:
\begin{corollary}\label{intersection}
If the bulk spacetime is smooth, and if $\gamma_{1,2}$ are HRT surfaces for $\mathsf{A}$, then they intersect only on coincident connected components. In other words, any connected component of $\gamma_1$ and any connected component of $\gamma_2$ are either equal or non-intersecting.
\end{corollary}
\begin{proof}
First we will show that, if $\gamma_{1,2}$ intersect other than on entire connected components, then they must intersect transversely (i.e.\ non-tangently) somewhere. Suppose they are tangent on a locus $L$ that does not include an entire connected component. Working in a maximin slice $\sigma$, and using Gauss normal coordinates adapted to $\gamma_1$, the position of $\gamma_2$ in a neighborhood of $L$ obeys Laplace's equation. This is impossible unless $L$ has codimension at least 2 (with respect to $\sigma$), in which case it has codimension-1 zeros with nonzero gradient, implying $\gamma_{1,2}$ intersect transversely.

It follows that the HRT surfaces $\tilde\gamma_{1,2}$ constructed in lemma \ref{lemma:nested} (from the union and intersection of their entanglement wedges) have kinks along the intersection locus. But this is impossible in a smooth spacetime.
\end{proof}
This corollary can also be considered a corollary of theorem 17(d) of \cite{Wall:2012uf}, which states that the HRT surfaces of nested boundary regions do not intersect except on entire connected components, where we regard $\mathsf{A}$ as nested with respect to itself.

\bibliographystyle{jhep}
\bibliography{references}

\end{document}